\documentclass[a4paper,12pt]{article}
\usepackage{amsmath,amsthm,amssymb}
\usepackage{subfigure}
\usepackage{verbatim}
\usepackage[english]{babel}
\usepackage{graphicx}
\graphicspath{ {images/} }

\newtheorem{theorem}{Theorem}

\newtheorem{lemma}{Lemma}

\newtheorem{proposition}{Proposition}

\newcommand{\be}{\begin{equation}}
\newcommand{\ee}{\end{equation}}
\newcommand{\bea}{\begin{eqnarray}}
\newcommand{\eea}{\end{eqnarray}}
\newcommand{\beann}{\begin{eqnarray*}}
\newcommand{\eeann}{\end{eqnarray*}}

\DeclareMathOperator{\tr}{Tr}

\newcommand{\expec}[1]{\left\langle #1 \right\rangle}

\DeclareMathOperator{\e}{e}

\title{\bf Long-lived mesoscopic entanglement between two damped infinite harmonic chains}

\author{F. Benatti$^{a,b}$, F. Carollo$^{a,b}$
R. Floreanini$^{b}$, J. Surace$^{c}$\\
\\
\small ${}^a$Dipartimento di Fisica, Universit\`a di Trieste, 
34151 Trieste, Italy\\
\small ${}^b$Istituto Nazionale di Fisica Nucleare, Sezione di Trieste,
34151 Trieste, Italy\\
\small ${}^c$Department of Physics and SUPA, University of Strathclyde,\\ 
\small Glasgow G4 0NG, United Kingdom}

\date{\null}

\begin{document}

\maketitle

\begin{abstract}
\noindent
We consider two chains, each made of $N$ independent oscillators, immersed in a common thermal
bath and study the dynamics 
of their mutual quantum correlations in the thermodynamic, large-$N$ limit. 
We show that dissipation and noise due to the presence of the external environment are able to generate collective quantum correlations between the two chains at the mesoscopic level. The created collective quantum entanglement
between the two many-body systems turns out to be rather robust, surviving for asymptotically long times
even for non vanishing bath temperatures.
\end{abstract}

\vskip 1cm

\section{Introduction}

Many-body systems are quantum systems composed by a large number $N$ of elementary
constituents. Because of the multiplicity of involved elements,
the study of single particle properties is impractical;
physically measurable properties of such systems are instead collective observables,
{\it i.e.} observable involving all system degrees of freedom.

In general, such collective observables represent extensive properties
of the systems, growing indefinitely with the number $N$. Collective observables
need therefore to be normalized by suitable powers of $1/N$. Provided the
system density $N/V$ is kept fixed, $V$ being the system volume, these normalized
observables become independent from the number $N$, allowing one to work in the
so-called thermodynamic, large-$N$ limit \cite{Feynman}-\cite{Bratteli}.

Typical examples of collective observables are {\it mean-field} observables,
{\it i.e.} averages over all constituents of single particle quantities,
an example of which is the mean magnetization in spin systems.
Although the single particle observables behave as quantum,
mean-field observables show in general a classical behaviour
as the number $N$ of constituents increases,
thus becoming {\it macroscopic} observables.
The well-established mean-field approximation theory precisely
describes many-body systems at this macroscopic level.

Nevertheless, recently there have been studies reporting the observation of
quantum behaviours also in systems made of a large number of particles \cite{Julsgaard}-\cite{Purdy}; 
typically, these systems either involve Bose-Einstein condensates, namely thousands of ultracold atoms 
trapped in optical lattices \cite{Leggett1}-\cite{Lewenstein2}, 
hybrid atom-photon \cite{Haroche}-\cite{Gerry} 
or optomechanical systems made of micro-oscillators (cantilevers) \cite{Wallquist}-\cite{Bowen}.

Clearly, mean-field observables, being averaged quantities, scaling as $1/N$ for large $N$,
can not be used to explain quantum effects on such scales. However, other kinds of collective observables
have been introduced and studied in many-body systems \cite{Goderis1}-\cite{Verbeure-book}; 
in analogy with classical probability theory, they are called fluctuations.
They still involve all the degrees of freedom of man-body quantum systems, but scale as $1/\sqrt{N}$
as the number $N$ of constituents increases, retaining some quantum properties
even in the thermodynamic limit. Being half-way between the microscopic observables,
those describing single particles behaviours, and the macroscopic mean-field
observables, they can be called {\it mesoscopic} observables.

One of the most striking manifestation of quantum coherence is entanglement,
{\it i.e.} the possibility of creating correlations within a bipartite system
that have no classical counterparts \cite{Horodecki}. 
Although treated at the beginning as a mere
theoretical curiosity, entanglement has became nowadays a real physical resource,
allowing achievements in quantum information and communication
not possible with purely classical means \cite{Nielsen,Petritis}.

Entanglement is however a fragile resource, that can be rapidly destroyed by the presence
of an external environment, which acts as a source of noise and dissipation, 
commonly leading to decoherence effects \cite{Alicki1}-\cite{Chruscinski1}.
Nevertheless, an external environment can not only degrade quantum coherence,
but also generate it through a purely mixing mechanism. Indeed, it has been shown that,
under certain circumstances, two independent, non interacting systems can become entangled by the action
of a common bath in which they are immersed \cite{Plenio1}-\cite{Tamura}. 
In general, the standard way for entangling 
two quantum systems is to make them directly interact; a different possibility is instead
to put them in contact with a common external bath: the presence of the bath may in fact induce
an indirect coupling between the two systems able to generate entanglement.

This possibility have been explored in microscopic systems, made of two-level atoms
or oscillators. In view of the recent developments in optomechanics and ultracold atom
experiments, it is of great interest to study whether two many-body systems can
similarly get entangled by the action of the environment in which they are immersed.
Notice that, being entanglement an intrinsically quantum phenomenon, 
this can possibly occur only at the mesoscopic level, {\it i.e.}
through collective observables that retain a quantum character even in the thermodynamic limit.

The present investigation contributes to answering this question: we shall see that two
non-interacting systems, made of a collection of independent oscillators, and immersed in
a common bath, can become entangled at the level of mesoscopic fluctuations
through a purely mixing mechanism. Even more strikingly, in certain situations,
the created entanglement can persist for asymptotically long times and nonvanishing temperatures.%
\footnote{
Similar issues have been previously investigated in the case of spin chains, involving
finite dimensional algebras at each site \cite{Carollo1}-\cite{Carollo6}: the generalization to the case of oscillators
is non-trivial and requires the use of quite different mathematical tools.}

The next Section is dedicated to the theory of collective observables in many-body systems. 
Referring to the investigated system, we shall first discuss mean-field observables and analyze their classical behaviour in the thermodynamic limit. Then, we introduce and study
suitable fluctuation operators and discuss the so-called {\it quantum central limit} \cite{Verbeure-book}:
it allows to assimilate at the mesoscopic level fluctuations operators to suitable bosonic variables.

Our double chain system is assumed to be immersed in and weakly coupled to a external bath; its resulting open
dynamics is discussed in Section 3. It can be described in terms of a master equation of very general form,
expressed in terms of microscopic variables: it connects indirectly the two, otherwise independent chains. 
By a careful choice of fluctuation operators, it is then shown that the microscopic open dynamics 
induces at the mesoscopic level a dissipative time-evolution
for the bosonic variables corresponding to these fluctuations. This dynamics
is of {\it quasi-free} type~\cite{Alicki2}, sending Gaussian states into Gaussian states.

We then focus on bosonic variables corresponding to the mesoscopic limit of
suitable fluctuation operators of either one or the other of the two chains. 
By examining the time evolution of these
sets of mesoscopic bosonic variables, in Section 4 we shall show that indeed the two chains can get entangled
by the action of the bath. How the amount of generated entanglement depends on the bath temperature
and system-bath coupling parameter will also be discussed in detail. Remarkably, in certain situations, 
entanglement can persist for asymptotically long times even at nonvanishing temperature.

Finally, the Appendix contains proofs and technical calculations that can not be
accommodated in the main text.

\section{Many-body systems at the mesoscopic scale}

In this section, we shall see that the common wisdom that assigns a ``classical'' behaviour
to observable averages while a non-trivial dynamics to observable fluctuations holds also in the case
of quantum many-body systems. More specifically, mean-field observables will be shown to provide 
a classical (commutative) description of the system, typical of the ``macroscopic'' world, 
while fluctuations around observable averages still retain some quantum (noncommutative) properties: 
they describe the ``mesoscopic'' behaviour of the system, at a level that is half way between the
microscopic and macroscopic scales.

\subsection{Many-body oscillator system}

We shall study the dynamics of a many-body system made of two equal, one-dimensional chains, 
each composed of $N$ identical, independent oscillators.
Each site $k$ of the double chain system, $k=1,2,\ldots, N$, consists of a couple of harmonic oscillators:
they are described by the corresponding position $x^{[k]}_\alpha$ and momentum $p^{[k]}_\alpha$ operators, the index
$\alpha=1,2$ labelling the two chains; they obey standard canonical commutation relations:
\begin{equation}
\Big[ x^{[j]}_\alpha,\ p^{[k]}_\beta \Big]= i\,\delta_{jk}\ \delta_{\alpha\beta}\ ,\qquad j,k=1,2,\ldots,N\ ,\quad
\alpha, \beta=1,2\ .
\label{1.1}
\end{equation}
The oscillators are free and therefore their independent microscopic dynamics
is generated by the Hamiltonian:
\begin{equation}
H^{[k]}_\alpha = \frac{\omega}{2} \Big[ \big(x^{[k]}_\alpha\big)^2 + \big(p^{[k]}_\alpha\big)^2 \Big]\ ,
\label{1.2}
\end{equation}
with $\omega$ the oscillator frequency, taken to be the same for all sites. 
The system observables at site $k$ turn out to be
polynomials in the four variables $(x_1^{[k]},\ p_1^{[k]},\ x_2^{[k]},\ p_2^{[k]} )$, of which
the above Hamiltonian is just one example; recalling $(\ref{1.1})$,
all these polynomials form an algebra ${\cal A}^{[k]}$, that it is usually called the (double) oscillator algebra.
One can now take the union of all these algebras for $k=1, 2,\ldots, N$ and form the 
total algebra ${\cal A}^{(N)}$ containing the observables of the whole many-body system;
elements of ${\cal A}^{(N)}$ are polynomials in the $4N$ variables
$(x_\alpha^{[k]},\ p_\beta^{[k]})$, $k=1,2,\ldots, N$, $\alpha,\beta=1,2$. 
In particular, any element $O^{[k]}\in {\cal A}^{[k]}$ at site $k$ readily extends to an operator acting
on the whole system by simply making it act trivially on all sites except $k$.

As mentioned in the Introduction, the double chain system is assumed to be immersed
in a thermal bath; this is the most realistic situation encountered in actual experiments.
It is then reasonable to assume the system to be initially at thermal equilibrium
at the bath temperature $T\equiv 1/\beta$. The state of the double chain system can then
be described by the following density operator (Gibbs state):
\begin{equation}
\rho^{(N)}=\frac{e^{-\beta H^{(N)}}}{\tr\big[e^{-\beta H^{(N)}}\big]}\ ,
\label{1.3}
\end{equation}
where $H^{(N)}$ is the total Hamiltonian of the system, built from the single site ones in (\ref{1.2}):
\begin{equation}
H^{(N)}= \sum_{k=1}^N \sum_{\alpha=1}^2 H^{[k]}_\alpha\ .
\label{1.4}
\end{equation}
Since we are dealing with independent oscillators, the density matrix (\ref{1.3}) turns out to be the
product of single site density matrices:
\begin{equation}
\rho^{(N)}=\prod_{k=1}^N \rho^{[k]}\ ,\qquad 
\rho^{[k]}=\frac{e^{-\beta H^{[k]}}}{\tr\big[e^{-\beta H^{[k]}}\big]}\ ,\quad
H^{[k]}=\sum_{\alpha=1}^2 H^{[k]}_\alpha\ ,\quad k=1,2,\ldots, N\ .
\label{1.5}
\end{equation}
Given any observable $O$ of the many-body system, $O\in {\cal A}^{(N)}$, its expectation
value in the state (\ref{1.3}) can now be readily computed:
\begin{equation}
\langle O\rangle_N\equiv \tr\big[\rho^{(N)}\, O\big]\ .
\label{1.6}
\end{equation}
The thermal state (\ref{1.3}) possesses interesting properties. First, due to the translation 
invariance of the density operator $\rho^{(N)}$, averages of the same observable at different sites 
coincide:
\begin{equation}
\langle O^{[j]}\rangle_N = \langle O^{[k]}\rangle_N, \qquad j,k=1,2,\ldots,N\ .
\label{1.7}
\end{equation}
Indeed, the two averages above reduce to $\langle O^{[j]}\rangle_N=\tr[\rho^{[j]} O^{[j]}]$
and $\langle O^{[k]}\rangle_N=\tr[\rho^{[k]} O^{[k]}]$, respectively, and since both single-site states 
$\rho^{[j]}$, $\rho^{[k]}$ and operators $O^{[j]}$, $O^{[k]}$ have the same dependence
on the corresponding canonical variables $(x_\alpha^{[j]},\ p_\alpha^{[j]})$ and 
$(x_\alpha^{[k]},\ p_\alpha^{[k]})$, the two averages coincide. In other terms,
the mean value of single site operators are independent both from the site index and the chain length $N$;
to remark this fact, in the following we shall use the simpler notation:
\begin{equation}
\langle O^{[k]}\rangle_N \equiv \langle O\rangle, \qquad k=1,2,\ldots,N\ .
\label{1.8}
\end{equation}
Similarly, one sees that in the state (\ref{1.3}) there are no correlation between different sites;
given two single site operators $A$ and $B$, one finds:
\begin{equation}
\big\langle A^{[j]}\, B^{[k]} \big\rangle_N - \big\langle A^{[j]}\big\rangle_N 
\big\langle B^{[k]}\big\rangle_N
=\Big(\langle A\, B\rangle - \langle A\rangle \langle B\rangle\Big)\, \delta_{jk}\ .
\label{1.9}
\end{equation}
Actually, any $n$-point correlation function can be expressed in terms of the above two-point ones,
since the state $\rho^{(N)}$ is Gaussian \cite{Ferraro}-\cite{Petz}. In order to appreciate this point, it is convenient
to introduce Weyl operators. Let us first collect the position and momentum operators
at the various $N$ sites into the $4N$-vector $\vec R$ with components \hfill\break
$(x_1^{[1]}, p_1^{[1]}, x_2^{[1]}, p_2^{[1]}, \dots, x_1^{[N]}, p_1^{[N]}, x_2^{[N]}, p_2^{[N]})$ 
and define the Weyl operator as:
\begin{equation}
\widehat{W}({\vec v}\,)=\e^{i \vec v \cdot \vec R}\ ,\qquad \vec v \cdot \vec R \equiv \sum_{i=1}^{4N} v_i\, R_i\ ,
\label{1.10}
\end{equation}
with $\vec v$ a $4N$-vector of real coefficients. The operators $\widehat{W}({\vec v}\,)$ are unitary,
$\big[\widehat{W}({\vec v}\,)\big]^\dagger=\widehat{W}(-{\vec v}\,)=\big[\widehat{W}({\vec v}\,)\big]^{-1}$ 
forming the so-called Weyl algebra ${\cal W}$,
characterized by the following relation, direct consequence of the canonical commutations in (\ref{1.1}):
\begin{equation}
\widehat{W}({\vec v}_1)\, \widehat{W}({\vec v}_2)=
\widehat{W}({\vec v}_1+{\vec v}_2)\e^{-\frac{i}{2} {\vec v}_1 \cdot \hat\sigma \cdot {\vec v}_2}\ ,
\label{1.11}
\end{equation}
with $\hat\sigma$ a $4N\times 4N$ symplectic matrix, that takes the following block-diagonal form
\begin{equation}
\hat\sigma=\begin{bmatrix} 
\begin{matrix}0 & 1\\ -1 & 0\end{matrix} & & 0 \\ & \ddots & \\ 0 & & \begin{matrix}0 & 1 \\ -1 & 0
\end{matrix} 
\end{bmatrix}.
\label{1.12}
\end{equation}
Any element of the oscillator algebra ${\cal A}^{(N)}$ can be obtained by taking
derivatives of $\widehat{W}({\vec v}\,)$ with respects of the components of the coefficient vector $\vec v$,
so that the description of the system in terms of Weyl operators is physically
equivalent to that in terms of the canonical variables. Nevertheless, it is
preferable to deal with Weyl operators, since these are bounded operators, unlike
coordinate and momentum operators. Indeed the oscillator algebra $\mathcal{A}^{(N)}$
should be really identified with the closure of the Weyl algebra with respect to the so-called GNS-representation
based on the chosen state $\rho^{(N)}$ (for details, see \cite{Strocchi5,Thirring1,Bratteli}).
In this way, the algebra $\mathcal{A}^{(N)}$ contains only bounded operators; in the following,
when referring to the oscillator algebra, we will always mean the algebra $\mathcal{A}^{(N)}$
constructed in this way.

A state $\rho$ for the system is called Gaussian if the expectation of Weyl operators 
are in Gaussian form, namely:
\begin{equation}
\tr \big[ \rho\, \widehat{W}({\vec v}\,)\big]=\e^{-\frac{1}{2} ({\vec v} \cdot \Sigma \cdot {\vec v})}\ ,
\label{1.13}
\end{equation}
where $\Sigma$ is the covariance matrix, whose components $[\Sigma]_{ij}$ are defined through the
anticommutator of the different components $R_i$ of $\vec R$:
\begin{equation}
[\Sigma]_{ij}\equiv\frac{1}{2} \tr \Big[ \rho\, \{ R_i,\, R_j\} \Big]\ ,\qquad i,j=1,2,\ldots, 4N\ .
\label{1.14}
\end{equation}
For the thermal state $\rho^{(N)}$ in (\ref{1.3}), the covariance matrix is explicitly given by:
\begin{equation}
\Sigma^{(N)}=\frac{1}{2\eta}\, {\bf 1}_{4N}\ ,\qquad \eta=\tanh\bigg(\frac{\beta\omega}{2}\bigg)\ ,
\label{1.15}
\end{equation}
where the notation ${\bf 1}_n$ indicates the identity matrix in $n$-dimensions.
Since the covariance matrix is proportional to the unit matrix, 
the state $\rho^{(N)}$ exhibits no correlations between oscillators
belonging to different chains, even at the same site $k$;
the state is therefore completely separable, as also explicitly exhibited by its
product form in (\ref{1.5}).

\subsection{Mean-field observables}

We have so far discussed single-site operators, the ones that are needed for a microscopic
description of the double chain system. However, due to experimental limitations, these operators are hardly 
accessible in practice; only, collective observables, involving all system sites, are in general
available to experimental investigations. 

In order to move from a microscopic description to a one involving collective operators, potentially
defined over an infinitely long double chain, a suitable scaling needs to be chosen.
The simplest example of collective observables are {\it mean-field} operators, 
{\it i.e.} averages of $N$ copies of a same single site observable $X$:
\begin{equation}
X^{(N)}=\frac{1}{N}\sum_{k=1}^N X^{[k]}\ ;
\label{1.16}
\end{equation}
 we are interested in
studying their behaviour in the thermodynamic, large-$N$ limit.

Let us then consider two of such operators $X^{(N)}$ and $Y^{(N)}$, constructed from single site observables $X$ and $Y$,
respectively, and compute their commutator:
\begin{equation}
\Big[X^{(N)},\, Y^{(N)}\Big]= \frac{1}{N^2}\sum_{j,k=1}^N \Big[X^{[j]},\, Y^{[k]}\Big]=
\frac{1}{N^2}\sum_{k=1}^N \Big[X^{[k]},\, Y^{[k]}\Big]\ ,
\label{1.17}
\end{equation}
where the last equality comes from the fact that operators belonging to different sites commute,
$\big[X^{[j]},\, Y^{[k]}\big]=\delta_{jk}\, Z^{[k]}$,
where $Z^{[k]}\equiv\big[X^{[k]},\, Y^{[k]}\big]$ is an operator at site $k$. One then realizes that
the commutator of two mean-field operators is still a mean-field operator:
\begin{equation}
\Big[X^{(N)},\, Y^{(N)}\Big]= \frac{1}{N}\,  Z^{(N)}\ ;
\label{1.18}
\end{equation}
however, due to the $1/N$ factor, it vanishes in the large-$N$ limit, in any topology where the limit
of $Z^{(N)}$ exists. In other terms,
mean-field operators can provide only a ``classical'' description of many-body systems,
any quantum, non-commutative character being lost in the thermodynamic limit.

The above result actually holds in the so-called {\it weak operator topology} \cite{Bratteli}, {\it i.e.} under state average
(see Section 6.1 in the Appendix for details); this means that for any local elements $A$ and $B$ in the algebra ${\cal A}^{(N)}$,
{\it i.e.} with support only on a finite number of sites,
one finds:
\begin{equation}
\lim_{N\to \infty} \langle A\, X^{(N)}\, B \rangle_N = \langle X \rangle\, \langle AB\rangle\ ;
\label{1.19}
\end{equation}
as a consequence, the large-$N$ limit of $X^{(N)}$ is a scalar multiple of the identity operator,
and we can write
\begin{equation}
\lim_{N\to \infty} X^{(N)}= \langle X \rangle\, {\bf 1}\ .
\label{1.20}
\end{equation}
Therefore, mean-field observables describe what we can call ``macroscopic'', classical degrees of freedom;
although constructed in terms of microscopic operators, in the large-$N$ limit they do not retain
any fingerprint of quantum behaviour. Instead, as remarked in the Introduction, we are interested in
studying collective observables, {\it} involving all system sites, and nevertheless showing a quantum
character even in the thermodynamic limit. Clearly, a less rapid scaling than $1/N$ of (\ref{1.16}) is needed.

\subsection{Fluctuations}

Fluctuation operators are collective observables that scale as the square root of $N$ and represent a deviation
from their averages. Given any single-site operator $X$, and its copies $X^{[k]}$ attached to the $k$-th 
site of the system, its corresponding
fluctuation operator $F^{(N)}(X)$ is defined as
\begin{equation}
F^{(N)}(X)\equiv \frac{1}{\sqrt N}\sum_{k=1}^N \Big(X^{[k]}-\langle X \rangle\Big)\ ;
\label{1.21}
\end{equation}
it is the quantum analog of the fluctuation of a random variable in classical probability theory.
In particular, note that its mean value vanishes: $\langle F^{(N)}(X) \rangle_N=\,0$.

Although the scaling $1/\sqrt N$ does not in general guarantee convergence, it is easy to show that
the fluctuation (\ref{1.21}) nevertheless retain a quantum behaviour in the large-$N$ limit. 
Indeed, let us consider the commutator
of two such fluctuations; following the same steps leading to (\ref{1.18}), one gets:
\begin{equation}
\Big[F^{(N)}(X),\, F^{(N)}(Y)\Big]= \frac{1}{N}\sum_{k=1}^N \Big[X^{[k]},\, Y^{[k]}\Big]\equiv Z^{(N)}\ ,
\label{1.22}
\end{equation}
with $Z^{(N)}$ a mean-field operator. Therefore, in the large-$N$ limit the commutator of the two fluctuations
tend to a scalar multiple of the identity operator, $\langle Z\rangle\, {\bf 1}$. This result indicates
that, focusing on fluctuation operators, a non-commutative, bosonic algebraic structure emerges.

In order to construct and study this algebra, it is convenient to restrict the discussion and focus on the following
single site, hermitian operators:
\begin{eqnarray}
\nonumber
&& X_1=\frac{\sqrt{\eta}}{2}(x_1^2-p_1^2)\ , \hskip 2cm   X_2=\frac{\sqrt{\eta}}{2} \left(x_1  p_1+p_1 x_1\right)\ ,\\
\label{1.23}
&& X_3=\frac{\sqrt{\eta}}{2}(x_2^2-p_2^2)\ ,\hskip 2cm  X_4=\frac{\sqrt{\eta}}{2}\left( x_2 p_2+p_2 x_2 \right)\ ,\\
&& X_5=\sqrt{\frac{2}{\eta}}\, (x_1x_2-p_1p_2)\ ,\hskip 1.1cm   X_6 = \sqrt{\frac{2}{\eta}}\, (x_1 p_2+p_1 x_2)\ ,
\nonumber
\end{eqnarray}
with $\eta$ as in (\ref{1.15}), and on their corresponding six fluctuation operators
$F^{(N)}(X_\mu)$,\break \hbox{$\mu=1,2,\ldots,6$}. 
Given the real linear span $\mathcal{X}$ generated by the operators $X_\mu$,
\begin{equation}
\mathcal{X}=\Big\{X_r\ \big|\ X_r\equiv\vec{r}\cdot\vec{X}=\sum_{\mu=1}^6 r_\mu\, X_\mu,\ \vec{r}\in\mathbb{R}^6\Big\}\ ,
\label{1.24}
\end{equation}
one can further consider the fluctuations of the combination $X_r$, 
which, by the linearity of the definition (\ref{1.21}) assigning to any single-site operator its corresponding fluctuation,
can be rewritten as:
\begin{equation}
F^{(N)}(X_r)=\sum_{\mu=1}^6 r_\mu\, F^{(N)}(X_\mu)\equiv \vec{r}\cdot\vec{F}^{(N)}(X)\ .
\label{1.25}
\end{equation}
Let us then study the large $N$ behaviour of the Weyl-like operator
\begin{equation}
W^{(N)}(\vec{r}\,)=\e^{i \vec{r}\cdot\vec{F}^{(N)}(X) }\ ,
\label{1.26}
\end{equation}
which, unlike the fluctuations $F^{(N)}(X_\mu)$, is a bounded operator.
The large-$N$ limit of the average of $W^{(N)}(\vec{r}\,)$ in the chosen system state (\ref{1.3})
turns out to be of Gaussian form:

\begin{lemma}
Given the state $\rho^{(N)}$ in (\ref{1.3}), and the Weyl-like operator $W^{(N)}(\vec{r}\,)$ in (\ref{1.26}), one has:
$$
\lim_{N\to\infty}\big\langle W^{(N)}(\vec{r}\,)\big\rangle_N=e^{-\frac{1}{2}\vec{r}\cdot\Sigma_\beta\cdot \vec{r}}\, ,
$$
with 
\begin{equation}
\Sigma_\beta=\frac{\eta^2+1}{4\eta}\,{\bf 1}_{6}\, .
\label{1.27}
\end{equation}
\end{lemma}

\begin{proof}
Let us consider the expectation $\langle W^{(N)}(\vec{r}\, )\big\rangle_N$; recalling (\ref{1.7}) 
and expanding the exponential, one can write:
\begin{eqnarray}
\nonumber
\langle W^{(N)}(\vec{r}\,)\big\rangle_N &=& 
\expec{\prod_{k=1}^N \e^{\frac{i}{\sqrt{N}}\big(X_r^{[k]}-\expec{X_r}\big)}}_N=
\left(\expec{e^{\frac{i}{\sqrt{N}}\big(X_r-\expec{X_r})} }\right)^N=\\
&=&\Bigg(1+\frac{i}{\sqrt{N}}\big\langle{X_r-\expec{X_r}}\big\rangle 
+\frac{i^2}{2N}\expec{X_r^2-\expec{X_r}^2}+\expec{{\cal R}_r^{(N)}}\Bigg)^N\, ,
\label{1.28}
\end{eqnarray}
with 
\begin{equation}
{\cal R}_r^{(N)}=\sum_{k=3}^\infty\frac{i^k}{k!\,\big(\sqrt{N}\big)^k}\big(X_r-\expec{X_r}\big)^k\, .
\label{1.29}
\end{equation}
Using the results of {\sl Lemma 3} in Section 6.2 of the Appendix, 
one sees that for large $N$ the contribution ${\cal R}_r^{(N)}$
is such that
$$
\big|\!\expec{{\cal R}_r^{(N)}}\!\big|= O\left(N^{-3/2}\right)\ ,
$$
and therefore it is subdominant with respect to the other pieces in (\ref{1.28}).
As a result, 
$$
\lim_{N\to\infty}\big\langle W^{(N)}(\vec{r}\,)\big\rangle_N=
\lim_{N\to\infty}\left(1-\frac{\expec{X_r^2-\expec{X_r}^2}}{2N}\right)^N=
\e^{-\frac{1}{2}\vec{r}\cdot\Sigma_\beta\cdot\vec{r}}\, ,
$$
with the covariance matrix $\Sigma_\beta$ defined through the following identity:
\begin{equation}
\expec{X_r^2-\expec{X_r}^2}=\vec{r}\cdot\Sigma_\beta\cdot \vec{r}\ ,
\label{1.30}
\end{equation}
and direct evaluation gives (\ref{1.27}).
\end{proof}

Recalling the result (\ref{1.13}), the statement of this {\sl Lemma} suggests that in the large-$N$ limit 
the operators $W^{(N)}(\vec{r}\, )$ behave as true Weyl operators. In order to confirm this,
the analog of the algebraic relation (\ref{1.11}) should be proven; this is precisely the content of
the following {\sl Lemma}, whose proof can be found in Section 6.3 of the Appendix.

\begin{lemma}
Given the state $\rho^{(N)}$ in (\ref{1.3}), and two distinct single-site operators $X_{r_1}$, $X_{r_2}$ belonging
to the real linear span $\mathcal{X}$ in (\ref{1.24}), one has
$$
\lim_{N\to\infty}\expec{A^{(N)}\,\left(W^{(N)}(\vec{r_1})\, W^{(N)}(\vec{r_2})-
W^{(N)}(\vec{r_1}+\vec{r_2})\ \e^{-\frac{1}{2}\expec{\big[X_{r_1}, X_{r_2}\big]}}\right)}_N=0\ ,
$$
for any (bounded) element $A^{(N)}$
in the oscillator algebra ${\cal A}^{(N)}$.
\label{Weyl}
\end{lemma}
\noindent
Notice that the single-site expectation $\big\langle[X_{r_1},X_{r_2}]\big\rangle$ can be easily computed:
\begin{equation}
\big\langle[X_{r_1},X_{r_2}]\big\rangle=\sum_{\mu\nu=1}^{6}(r_1)_\mu\, 
\big\langle\big[X_\mu,\, X_\nu \big]\big\rangle\,(r_2)_\nu\equiv
i\, \vec{r_1}\cdot
\sigma\cdot\vec{r_2}\ ,
\label{1.32}
\end{equation}
with the $6\times 6$ antisymmetric matrix $\sigma$ explicitly given by:
\begin{equation}
\big[\sigma\big]_{\mu\nu}=-i\big\langle\big[X_\mu,\, X_\nu \big]\big\rangle=
\begin{pmatrix}
 i\sigma_2&{ 0}&{ 0}\\
{0}&i\sigma_2& { 0}\\
{ 0}&{ 0}&i\sigma_2
\end{pmatrix}\ ,
\label{1.33}
\end{equation}
$\sigma_2$ being the standard second Pauli matrix. The real vector space $\mathcal{X}$ in (\ref{1.24}) is then endowed
with the symplectic matrix $\sigma$ and we can define on it the abstract Weyl algebra 
${\cal W}(\mathcal{X},\sigma)$, linearly generated by the Weyl operators $W(\vec{r}\,)$,
${\vec r}\in \mathbb{R}^6$, obeying the defining relations (compare with (\ref{1.11})):
\begin{equation}
\begin{split}
&W(\vec{r}_1)\, W(\vec{r}_2)=W(\vec{r}_1+\vec{r}_2)\ e^{-\frac{i}{2}\vec{r}_1\cdot\sigma\cdot\vec{r}_2}\ ,\\
&W^\dagger(\vec{r}\,)=W(-\vec{r}\,)\ .
\end{split}
\label{1.34}
\end{equation}
We can then say that in the large-$N$ limit the Weyl-like operator $W^{(N)}(\vec{r}\,)$ in (\ref{1.26})
yield the true Weyl operator $W(\vec{r}\,)$, element of the algebra ${\cal W}(\mathcal{X},\sigma)$.
The precise way in which this statement should be understood is provided by the following theorem:

\begin{theorem}
\label{th1}
Given the state $\rho^{(N)}$ and the real linear vector space $\mathcal{X}$ generated by 
the operators $X_\mu$ in (\ref{1.23}), one can define a Gaussian state $\Omega$ on the Weyl algebra
${\cal W}(\mathcal{X},\sigma)$ such that, for all $\vec{r}_i\in\mathbb{R}^6$, $i=1,2,\ldots,n$,
\begin{equation}
\lim_{N\to\infty}\Big\langle W^{(N)}(\vec{r}_1)\,W^{(N)}(\vec{r}_2)\,\cdots W^{(N)}(\vec{r}_n)\Big\rangle_N=
\Big\langle W(\vec{r}_1)\,W(\vec{r}_2)\,\cdots W(\vec{r}_n)\Big\rangle_\Omega\ ,
\label{1.35}
\end{equation}
with
\begin{equation}
\lim_{N\to\infty}\Big\langle W^{(N)}(\vec{r}\,)\Big\rangle_N=
\e^{-\frac{1}{2} \vec{r}\cdot\Sigma_\beta\cdot \vec{r}}=
\Big\langle W(\vec{r}\,)\Big\rangle_\Omega\ ,\qquad \forall\, \vec{r}\in\mathbb{R}^6\ .
\label{1.36}
\end{equation}
\end{theorem}
\medskip

\begin{proof}
Both limits are direct consequences of the previous two Lemmas. What it is left to check is to
ensure that the Gaussian state $\Omega$ on the algebra $\mathcal{W}(\mathcal{X},\sigma)$ 
is indeed a well defined state. First of all, it is normalized as easily seen by
setting $\vec{r}=0$ in (\ref{1.36}). Further, its positivity is guaranteed
by the following inequality connecting covariance and symplectic matrices \cite{Holevo}:
$$
\Sigma_\beta+\frac{i}{2}\sigma\ge 0\ .
$$
\end{proof}

Since $\Omega$ is a Gaussian state, it provides a regular representation \cite{Bratteli} 
of the Weyl algebra $\mathcal{W}(\mathcal{X},\sigma)$; this guarantees that one can write:
\begin{equation}
W(\vec{r}\,)={\rm e}^{i\,\vec{r}\cdot\vec{F}}\ ,\qquad 
\vec{r}\cdot\vec{F}=\sum_{\mu=1}^6\,r_\mu\,F_\mu\ ,
\label{1.37}
\end{equation}
where $F_\mu$ are collective field operators satisfying canonical commutation relations:
\begin{equation}
\left[\vec{r}_1\cdot\vec{F}\,,\,\vec{r}_2\cdot\vec{F}\right]=
i\,\vec{r}_1\cdot\sigma\cdot\vec{r}_2\ .
\label{1.38}
\end{equation}
Then, through (\ref{1.26}) and (\ref{1.36}), {\it i.e.}  
$\lim_{N\to\infty}\big\langle\e^{i \vec{r}\cdot\vec{F}^{(N)}(X) }\big\rangle_N= 
\big\langle\e^{i \vec{r}\cdot\vec{F}}\big\rangle_\Omega$,  one can identify
\begin{equation}
\lim_{N\to\infty}F^{(N)}(X_\mu)=F_\mu\ ,\qquad \mu=1,2,\ldots,6\ .
\label{1.39}
\end{equation}
Further, in view of the explicit form (\ref{1.33}) of the symplectic matrix $\sigma$, the components $F_\mu$
can be labelled as
\begin{equation}
\vec{F}=(\hat X_1,\hat P_1,\hat X_2,\hat P_2,\hat X_3,\hat P_3)\ ,
\label{1.40}
\end{equation}
with the $\hat X_i$ position- and $\hat P_i$ momentum-like operators, satisfying
$$
\left[\hat X_i,\hat P_j\right]=i\delta_{i,j}\ ,\qquad i,j=1,2,3\ ,
$$
as a consequence of (\ref{1.38}) above.
Recalling the definitions (\ref{1.23}), one sees that the couple $\hat X_1$, $\hat P_1$ are operators
pertaining to the first chain of oscillators, while $\hat X_2$, $\hat P_2$ to the second one. On the contrary,
$\hat X_3$, $\hat P_3$ are mixed operators belonging to both chains.
Further, one can show that any other single-site oscillator operator not belonging
to the linear span $\mathcal{X}$ give rise to fluctuation operators that in the large-$N$ limit
dynamically decouple from the six in (\ref{1.40}) (see later and \cite{Surace}); 
this is why we can limit the discussion to the chosen set (\ref{1.23}).

Notice that the state $\Omega$ is separable with respect to the three modes (\ref{1.40}):
its covariance matrix $\Sigma_\beta$ is diagonal, thus showing neither quantum nor classical
correlations.
Indeed, as in the case of $\rho^{(N)}$, the state $\Omega$ can be represented by a density matrix $\rho_\Omega$ 
in product form, $\rho_\Omega=\prod_{i=1}^3 \rho_\Omega^{(i)}$, with
$\rho_\Omega^{(i)}$ standard free oscillator Gaussian states
in the variables $\hat X_i$ and $\hat P_i$.

It should be stressed that the bosonic canonical variables $\hat X_i$, $\hat P_i$, $i=1,2,3$,
are collective operators, originating from the fluctuation operators $F^{(N)}(X_\mu)$
through the limiting procedure (\ref{1.39}), specified by the previous {\sl Theorem 1}.
They describe the behaviour of the double chain system at a level that is half way
between the microscopic world of single-site oscillators and the macroscopic realm
of mean-field observables.
At this intermediate, mesoscopic level, a collective quantum behaviour of the many-body system is still
permitted.

In this respect, the large-$N$ limit that allows to pass from the exponential 
(\ref{1.26}) of local fluctuations (\ref{1.21})
to the mesoscopic operators belonging to the Weyl algebra $\mathcal{W}(\mathcal{X},\sigma)$
can be called the {\it mesoscopic limit}. Indeed, observe that
by varying $\vec{r}_{1},\ \vec{r}_{2}\in \mathbb{R}^6$, the expectation values of the form
$\big\langle W(\vec{r}_1)\,O\,W(\vec{r}_2)\big\rangle_\Omega$ completely determine any
generic operator $O$ in the Weyl algebra $\mathcal{W}(\mathcal{X},\sigma)$: essentially,
they represent its corresponding matrix elements.%
\footnote{In more precise mathematical terms, the r.h.s of (\ref{1.41}) corresponds to the
matrix elements of the operator $\pi_\Omega(O)$ with respect to the two vectors
$\pi_\Omega\big(W(\vec{r}_1)\big) |\Omega\rangle$ and $\pi_\Omega\big(W(\vec{r}_2)\big) |\Omega\rangle$ 
in the GNS-representation of the Weyl algebra
$\mathcal{W}(\mathcal{X},\sigma)$ based on the state $\Omega$ \cite{Bratteli}. Since these vectors are dense in
the corresponding Hilbert space, those matrix elements completely define the operators $O$.}
Then, the convergence to $O$ of a sequence $O^{(N)}$ of linear combinations
of exponential operators $W^{(N)}(\vec{r}\,)$ can be given the following formal definition:

\bigskip\noindent
\textbf{Mesoscopic limit.} {\it 
Given a sequence of operators $O^{(N)}$, linear combinations of exponential operators $W^{(N)}(\vec{r}\,)$, 
we shall say that it possesses the mesoscopic limit $O$, writing
$$
m-\lim_{N\to\infty} O^{(N)}= O\ ,
$$
if and only if 
\begin{equation}
\lim_{N\to\infty} \Big\langle W^{(N)}(\vec{r}_1)\,O^{(N)}\,W^{(N)}(\vec{r}_2)\Big\rangle_N=
\Big\langle W(\vec{r}_1)\,O\,W(\vec{r}_2)\Big\rangle_\Omega\ ,
\label{1.41}
\end{equation}
for all $\vec{r}_{1},\ \vec{r}_{2}\in \mathbb{R}^6$.}
\bigskip

In the following, we shall examine the open dynamics of the double chain system; more precisely,
given a one-parameter family of microscopic dynamical maps $\Phi^{(N)}_t$ on the algebra ${\cal A}^{(N)}$,
we will study its action on the Weyl-like operators $W^{(N)}(\vec{r}\,)$,
the exponential of local fluctuations, in the limit of large $N$. In other terms, 
we shall look for the limiting {\it mesoscopic dynamics} $\Phi_t$ acting on the elements $W(\vec{r}\,)$ of the
Weyl algebra $\mathcal{W}(\mathcal{X},\sigma)$. In line with the previously introduced mesoscopic limit,
we can state the following definition:

\bigskip\noindent
\textbf{Mesoscopic dynamics.} {\it Given a family of one-parameter maps
$\Phi^{(N)}_t\, : {\cal A}^{(N)} \to {\cal A}^{(N)}$, we shall say that it gives the
mesoscopic limit $\Phi_t$ on the Weyl algebra $\mathcal{W}(\mathcal{X},\sigma)$,
$$
m-\lim_{N\to\infty} \Phi^{(N)}_t = \Phi_t\ ,
$$
if and only if
\begin{equation}
\lim_{N\to\infty} \Big\langle W^{(N)}(\vec{r}_1)\,\Phi^{(N)}_t\big[W^{(N)}(\vec{r}\,)\big]\, 
W^{(N)}(\vec{r}_2)\Big\rangle_N=
\Big\langle W(\vec{r}_1)\,\Phi_t\big[W(\vec{r}\,)\big]\, W(\vec{r}_2)\Big\rangle_\Omega\ ,
\label{1.42}
\end{equation}
for all $\vec{r},\ \vec{r}_{1},\ \vec{r}_{2}\in \mathbb{R}^6$.}

\section{Open system dynamics}

As mentioned in the Introduction, the double-chain system of oscillators is assumed to be immersed
in an external bath: this is the most common situation encountered in actual experiments performed
on many-body systems
that can never be thought of as completely isolated from their surroundings.
Because of the presence of the bath, the microscopic system dynamics can not be generated by
the oscillator Hamiltonian $H^{(N)}$ alone, as given in (\ref{1.4}); additional pieces accounting for
the dissipative and noisy effects induced by the environment ought to be present.

\subsection{Microscopic dissipative dynamics}

For a weakly coupled bath, standard techniques allow to obtain the master equation
generating the open dynamics of any microscopic observable $X$; 
it takes the Kossakowski-Lindblad form \cite{Alicki1}-\cite{Chruscinski1}:
\begin{equation}
\frac{d}{dt}X_t=\mathbb{L}^{(N)}\left[X_t\right]=i\left[H^{(N)},X\, _t\right]+\mathbb{D}^{(N)}\left[X_t\right]\ .
\label{2.1}
\end{equation}
Assuming the same bath coupling for all sites, we shall consider the dissipative part of the generator 
$\mathbb{L}^{(N)}$ of the following form:
\begin{equation}
\mathbb{D}^{(N)}\left[X\right]\equiv\sum_{k=1}^N \mathbb{D}^{[k]}[X]=\sum_{k=1}^N\sum_{\alpha,\beta=1}^{4} C_{\alpha\beta} 
\left( V^{[k]}_{\alpha} X V^{[k]}_{\beta}-\frac{1}{2}\left\{V^{[k]}_{\alpha}V^{[k]}_{\beta},\, X \right\}\right)\ ,
\label{2.2}
\end{equation}
where with $V^{[k]}$ we indicate the microscopic, site-$k$ operator-valued 
four-vector with components $(x_1^{[k]},\ p_1^{[k]},\ x_2^{[k]},\ p_2^{[k]} )$; the $4\times 4$
Kossakowski matrix $C$ with elements $C_{\alpha\beta}$ encodes the bath noisy properties and can be taken of the form:%
\footnote{The dissipative generator in (\ref{2.2}), with $C$ as (\ref{2.3}), (\ref{2.4}), is rather general 
and can be obtained through standard weak-coupling techniques \cite{Alicki1} starting from a microscopic 
system-environment interaction Hamiltonian of the form
$\sum_{k=1}^N\sum_{\alpha,\beta=1}^4 V_\alpha^{[k]} \otimes B_\alpha^{[k]}$, with
$B_\alpha^{[k]}$ suitable hermitian bath operators.}
\begin{equation}
{C}=
\left(
\begin{array}{c|c}
\mathbb{A}  & \mathbb{B} \\ \hline
 \mathbb{B}^\dagger & \mathbb{A}
\end{array}
\right),
\label{2.3}
\end{equation}
with
\begin{equation}
\mathbb{A}=\frac{1+\gamma}{2}
\begin{pmatrix}
1 & i\eta \\
-i\eta & 1
\end{pmatrix}\ ,
\qquad
\mathbb{B}=\lambda\, \mathbb{A}\ ,
\qquad
\gamma=\e^{-\beta\omega}\ ,\quad \eta=\tanh(\beta\omega/2)\ .
\label{2.4}
\end{equation}
The parameters $\gamma$ and $\eta$ contains the dependence on the bath temperature, while
$\lambda$ is a real constant that measures the bath induced statistical coupling between the two
chains of oscillators. The condition of complete positivity on the generated dynamics
requires $C$ to be positive semidefinite, which in turn gives
$\lambda^2\leq 1$.%
\footnote{Complete positivity is a condition more restrictive than simple positivity:
it needs to be enforced on any linear open dynamics in order for it to be physically consistent in
all situations; for more information and details, see \cite{Benatti-rev, Benatti-book}.}
The master equation (\ref{2.1}) with $\mathbb{D}^{(N)}$ as in (\ref{2.2}) generates a
one-parameter family of transformations mapping Gaussian states 
into Gaussian states \cite{Vanheuverzwijn,Benatti3}.
 
To appreciate the physical meaning of (\ref{2.2}), notice that 
the first two entries in ${V}^{[k]}$ refer to variables pertaining to the first chain,
while the remaining two to the second chain, so that the diagonal blocks ${\mathbb{A}}$ of the Kossakowski matrix
describe the evolution of the two chains independently interacting with the same bath; 
in absence of ${\mathbb{B}}$, the dynamics of the binary system would then be in product form.
Instead, the off-diagonal blocks ${\mathbb{B}}$ statistically couple the two chains,
and the strength of this coupling is essentially measured by the parameter $\lambda$.

Recalling the form of the Hamiltonian (\ref{1.4}) and that of $\mathbb{D}^{(N)}$ above, the dynamical generator
$\mathbb{L}^{(N)}$ in (\ref{2.1}) can be decomposed as
\begin{equation}
\mathbb{L}^{(N)}[X]=\sum_{k=1}^N \mathbb{L}^{[k]}[X]\equiv
\sum_{k=1}^N\bigg(i\left[H^{[k]},X\right]+\mathbb{D}^{[k]}\left[X\right]\bigg)\ ,
\label{2.6}
\end{equation}
where $\mathbb{L}^{[k]}$ acts only on site $k$. As a consequence, the dynamical
map $\Phi^{(N)}_t$ implementing the finite time evolution, formally
obtained from the master equation (\ref{2.1}) through exponentiation,
$\Phi^{(N)}_t=\e^{t\mathbb{L}^{(N)}}$,
does not create correlations between different sites; in other terms, given any system observable
in product form, $X=\prod_{k=1}^N X^{[k]}$, one has:
\begin{equation}
\e^{t\mathbb{L}^{(N)}}\left[\prod_{k=1}^N X^{[k]}\right]=\prod_{k=1}^N\, 
\e^{t\mathbb{L}^{[k]}}\left[X^{[k]}\right]\ .
\label{2.7}
\end{equation}
Further, one finds that the unitary dynamics generated by the Hamiltonian $H^{(N)}$ alone commutes
with the dissipative one generated by $\mathbb{D}^{(N)}$. Indeed, one easily checks that 
$$
\e^{it H^{(N)}}\, \mathbb{D}^{(N)}\, [X]\, \e^{-it H^{(N)}}=
\mathbb{D}^{(N)}\left[\e^{it H^{(N)}}\, X\, \e^{-it H^{(N)}}\right]\ ,
$$
so that
\begin{equation*}
\Phi^{(N)}_t[X]\equiv\e^{t\mathbb{L}^{(N)}} [X]=\e^{it H^{(N)}}\, \left(\e^{t\mathbb{D}^{(N)} }[X]\right)\,\e^{it H^{(N)}}= 
\e^{t\mathbb{D}^{(N)}}\left[\e^{it H^{(N)}}\,X\, \e^{it H^{(N)}}\right]\, .
\end{equation*}
Due to the presence of the dissipative part (\ref{2.2}), the dynamical maps $\Phi^{(N)}_t$
no longer form a group, but a semigroup of transformations, satisfying a forward in time
composition law, typical of irreversible time-evolutions:
\begin{equation}
\Phi^{(N)}_t \circ \Phi^{(N)}_s = \Phi^{(N)}_s \circ \Phi^{(N)}_t =\Phi^{(N)}_{t+s}\ ,\qquad \forall s,\, t\geq 0\ .
\label{2.8}
\end{equation}
In addition, it is worth noting that, due to unitality, {\it i.e.} $\Phi^{(N)}_t[{\bf 1}]={\bf 1}$, and complete positivity, 
the maps $\Phi_t^{(N)}$ obey Schwartz-positivity:
\begin{equation}
\Phi^{(N)}_t\big[X^\dag X\big]\,\geq\,\Phi^{(N)}_t\big[X^\dag\big]\,\Phi^{(N)}_t\big[X\big]\ .
\label{2.8-1}
\end{equation}
The family of maps $\{\Phi^{(N)}_t\}_{t\geq0}$ then defines a {\it quantum dynamical semigroup} \cite{Alicki1}.

Finally, notice that the thermal equilibrium state $\rho^{(N)}$ in (\ref{1.3})
is time-invariant under the dynamics implemented by $\mathbb{L}^{(N)}$, {\it i.e.}
$$
\expec{e^{t\mathbb{L}^{(N)}}\left[X\right]}_N=\expec{X}_N\ ,\qquad \forall X\in {\cal A}^{(N)}\ ,
$$
since $\tr\big(\rho^{(N)} \mathbb{L}^{(N)}[X]\big)=\,0$, a result that can be
checked by direct computation;%
\footnote{Indeed, passing from the Heisenberg to the Schr\"odinger picture through the duality
relation $\tr\big(\rho^{(N)} \mathbb{L}^{(N)}[X]\big)=\tr\big(\widetilde{\mathbb{L}}^{(N)}[\rho^{(N)}]\, X\big)$,
one easily proves that: $\widetilde{\mathbb{L}}^{(N)}[\rho^{(N)}]=\,0$.}
actually, since $\rho^{(N)}$ commutes with the Hamiltonian $H^{(N)}$, it is separately invariant 
for both the unitary and dissipative part of the evolution. Using all these information,
we shall now investigate what kind of time evolution the microscopic dynamical maps $\Phi^{(N)}_t$
induce on fluctuations in the large-$N$ limit, {\it i.e.} at the mesoscopic level.

\subsection{Mesoscopic dissipative dynamics}

In this Section we shall show that the mesoscopic dynamics emerging from the large-$N$ limit
of the time evolution $\Phi^{(N)}_t$, as specified by (\ref{1.42}),
is again a dissipative semigroup of maps $\Phi_t$
on the Weyl algebra $\mathcal{W}(\mathcal{X},\sigma)$, transforming Weyl operators 
into Weyl operators. Maps of this kind are called {\it quasi-free} and their
generic form is as follows \cite{Petz}-\cite{Demoen}:
\begin{equation}
\Phi_t\big[W(\vec{r}\,)\big]=\e^{f_t(\vec{r})}\, W(\vec{r}_t)\ ,
\label{2.9}
\end{equation}
with given time-dependent prefactor and parameters $\vec{r}_t$. In the present case, one finds 
\hbox{($T$ represents matrix transposition)}:
\begin{equation}
\vec{r}_t= \mathcal{M}_t{}^T\cdot \vec{r}\ ,\qquad \mathcal{M}_t=\e^{t {\cal L}}\ ,
\label{2.10}
\end{equation}
where the $6\times 6$ matrix $\mathcal{L}$ gives the action of the Kossakowski-Lindblad generator $\mathbb{L}^{(N)}$
on the fluctuations $F^{(N)}(X_\mu)$ of the six single site operators introduced in (\ref{1.23}):
\begin{equation}
\mathbb{L}^{(N)}\big[\vec{r}\cdot\vec{F}^{(N)}(X)\big]=\vec{r}\cdot\mathcal{L}\cdot \vec{F}^{(N)}(X)\ ;
\label{2.11}
\end{equation}
explicitly, one finds:
\begin{equation}
\mathcal{L}=(\gamma-1)\, {\bf 1}_6 + 2\omega\, \sigma + 
\frac{(\gamma-1)\lambda}{\sqrt2}
\begin{pmatrix}
 0&{ 0}&{ \bf 1}_2\\
0&{ 0}&{ \bf 1}_2\\
{ \bf 1}_2&{ \bf 1}_2&0
\end{pmatrix}\ ,
\label{2.12}
\end{equation}
with $\sigma$ as in (\ref{1.33}). Instead, the exponent of the prefactor can be cast in the following form:
\begin{equation}
f_t(\vec{r}_t)=-\frac{1}{2}\, \vec{r}_t\cdot \mathcal{K}_t \cdot \vec{r}_t\ ,\qquad
\mathcal{K}_t=\Sigma_\beta - \mathcal{M}_t\cdot \Sigma_\beta\cdot\mathcal{M}_t{}^T \ ,
\label{2.13}
\end{equation}
where $\Sigma_\beta$ is the covariance matrix in (\ref{1.27}). With these definitions, one can
state the following result, whose proof can be found in Section 6.5 below:

\begin{theorem}
Given the state $\rho^{(N)}$ in (\ref{1.3}), the real linear vector space $\mathcal{X}$ generated by 
the operators $X_\mu$ in (\ref{1.23}) and the corresponding 
Weyl-like operators $W^{(N)}(\vec{r}\,)=\e^{i \vec{r}\cdot\vec{F}^{(N)}(X) }$,
evolving in time with the semigroup of maps $\Phi^{(N)}_t\equiv\e^{t\mathbb{L}^{(N)}}$, generated by
$\mathbb{L}^{(N)}$ in (\ref{2.1})-(\ref{2.4}),
the mesoscopic limit
$$
m-\lim_{N\to\infty}\Phi^{(N)}_t\Big[W^{(N)}(\vec{r}\,)\Big]=\Phi_t\left[W(\vec{r}\,)\right]\ ,
$$
defines a Gaussian quantum dynamical semigroup $\left\{\Phi_t\right\}_{t\ge0}$ on the Weyl algebra 
of fluctuations $\mathcal{W}\left(\mathcal{X},\sigma\right)$, explicitly given by (\ref{2.9})-(\ref{2.13}).
\end{theorem}

The mesoscopic evolution maps $\Phi_t$ are unital, {\it i.e.} they map the
identity operator into itself, as it can be easily checked by letting $\vec{r}=\,0$ in (\ref{2.9}).
In addition, they compose as a semigroup; indeed, for all $s,\ t\geq 0$,
\begin{eqnarray*}
\Phi_s \circ \Phi_t \big[ W(\vec{r}\,)\big]=&&\e^{-\frac{1}{2}\big(\vec{r}\cdot \mathcal{K}_t\cdot\vec{r}
+\vec{r}_t\cdot \mathcal{K}_s\cdot\vec{r}_t\big) }\ W\big((\vec{r}_t)_s\big)\\
=&& \e^{-\frac{1}{2}\big(\vec{r}\cdot \mathcal{K}_t\cdot\vec{r}
+\vec{r}\cdot \big(\mathcal{M}_t\cdot \mathcal{K}_s\cdot\mathcal{M}_t^T\big)\cdot\vec{r}\big) }\ W\big(\vec{r}_{t+s}\big)\\
=&&\e^{-\frac{1}{2}\, \vec{r}\cdot \mathcal{K}_{t+s}\cdot\vec{r}}\, W\big(\vec{r}_{t+s}\big)
=\Phi_{t+s}\Big[ W(\vec{r}\,) \Big]\ .
\end{eqnarray*}
Finally, the maps  $\Phi_t$ are completely positive, since the following condition is satisfied \cite{Demoen}:
\begin{equation}
\Sigma_\beta +\frac{i}{2}\, \sigma \geq
\mathcal{M}_t\cdot \Big( \Sigma_\beta +\frac{i}{2}\, \sigma\Big)\cdot\mathcal{M}_t{}^T\ .
\label{2.13-1}
\end{equation}
Indeed, in view of (\ref{1.39}), for any complex vector $\vec{r}\in\mathbb{C}^6$, one can write:
$$
\vec{r}^{\, *}\cdot\Big(\Sigma_\beta +\frac{i}{2}\, \sigma\Big)\cdot\vec{r}=
\lim_{N\to\infty}\omega\Big( F^{(N)}(\vec{r}^{\, *})\, F^{(N)}(\vec{r}\,) \Big)\ ,
$$
as the covariance $\Sigma_\beta$ and symplectic matrix $\sigma$ are the real and immaginary
part of the \hbox{large-$N$} limit of the correlation matrix 
$\omega\big( F^{(N)}(X_\mu)\, F^{(N)}(X_\nu) \big)$ \cite{Carollo3,Carollo6}. 
In addition, since $\omega$ is invariant under the
time evolution generated by $\mathbb{L}^{(N)}$, {\it i.e.} $\omega=\omega \circ \Phi^{(N)}_t$,
one further has:
\begin{eqnarray*}
&&\vec{r}^{\, *}\cdot\Big(\Sigma_\beta +\frac{i}{2}\, \sigma\Big)\cdot\vec{r}=
\lim_{N\to\infty}\omega\circ \Phi^{(N)}_t\Big[ F^{(N)}(\vec{r}^{\, *})\, F^{(N)}(\vec{r}\,) \Big]\\
&&\hskip 5cm \geq \lim_{N\to\infty}\omega\Big( \Phi^{(N)}_t\big[F^{(N)}(\vec{r}^{\, *})\big]\, \Phi^{(N)}_t\big[F^{(N)}(\vec{r}\,)\big] \Big)\ ,
\end{eqnarray*}
where the inequality is a consequence of Schwartz-positivity, see (\ref{2.8-1}). Recalling
(\ref{2.11}) and (\ref{2.10}), one finally writes:
$$
\lim_{N\to\infty}\omega\Big( \Phi^{(N)}_t\big[F^{(N)}(\vec{r}^{\, *})\big]\, \Phi^{(N)}_t\big[F^{(N)}(\vec{r}\,)\big] \Big)=
\lim_{N\to\infty}\omega\Big(F^{(N)}(\vec{r}_t{}^*)\,F^{(N)}(\vec{r}_t)\Big)=
\vec{r}_t{}^*\cdot\Big(\Sigma_\beta +\frac{i}{2}\sigma\Big)\cdot\vec{r}_t\ ,
$$
thus recovering (\ref{2.13-1}).

Due to unitality and complete positivity, 
also the maps $\Phi_t$ obey Schwartz-positivity:
\begin{equation}
\Phi_t\big[X^\dag X\big]\,\geq\,\Phi_t\big[X^\dag\big]\,\Phi_t\big[X\big]\ .
\label{2.13-2}
\end{equation}
Using this property and the unitarity of the Weyl operators $W(\vec{r}\,)$, one further finds:
\begin{equation*}
\left|{\rm e}^{f_t(\vec{r})}\right|=\big\|\Phi_t \big[ W(\vec{r}\,)\big]\big\|\leq \|W(\vec{r}\,)\|=1\ ,
\end{equation*}
as can also be directly checked, being $\mathcal{K}_t$ in (\ref{2.13}) a positive definite matrix.

\subsection{Gaussian states and entanglement}

The mesoscopic dissipative dynamics $\Phi_t$ obtained in the previous section is quasi-free 
as it maps Weyl operators into Weyl operators. One can then define a dual map $\widetilde{\Phi}_t$ acting
on any state $\rho$ on the Weyl algebra $\mathcal{W}\left(\mathcal{X},\sigma\right)$, by sending it into
$\rho_t=\widetilde{\Phi}_t[\rho]$, according to the duality relation
\begin{equation}
\tr\Big[\widetilde{\Phi}_t[\rho]\,  W(\vec r\,)\Big]=\tr\Big[\rho\;\Phi_t[W(\vec r\,)]\Big]\ .
\label{2.22}
\end{equation}
As already observed, useful states on $\mathcal{W}\left(\mathcal{X},\sigma\right)$ are Gaussian states
(with zero averages), $\rho_\Sigma$, which are characterized by a Gaussian expectation on Weyl operators:
\begin{equation}
\tr\Big[\rho_\Sigma\, W(\vec{r}\,)\Big]=\e^{-\frac{1}{2}(\vec{r}\cdot\Sigma\cdot \vec{r})}\ ,
\label{2.23}
\end{equation}
with
\begin{equation}
[\Sigma]_{\mu\nu}\equiv\frac{1}{2} \tr \Big[ \rho_\Sigma\, \big\{ F_\mu,\, F_\nu\big\} \Big]\ ,\qquad \mu,\ \nu=1,\ldots,6\ ,
\label{2.24}
\end{equation}
$\{F_\mu\}$ being the bosonic operators introduced in (\ref{1.40}),
corresponding to the large-$N$ fluctuations of the six single-site basic variables chosen in (\ref{1.23}).

These states are completely identified by their covariance matrix $\Sigma$; in particular, as already observed,
positivity of $\rho_\Sigma$ is equivalent to the following
condition \cite{Holevo}:
\begin{equation}
\Sigma+\frac{i}{2}\sigma\geq0\ ,
\label{2.25}
\end{equation}
with $\sigma$ the symplectic matrix in (\ref{1.33}).
One can easily verify that the map $\widetilde{\Phi}_t$ transform Gaussian states into Gaussian states:
\begin{equation}
\nonumber
\tr\Big[\widetilde{\Phi}_t[\rho_\Sigma]\, W(\vec{r}\,)\Big]=
{\rm e}^{f_r(t)}\, \tr\Big[\rho_\Sigma\, W(\vec{r}_t\,)\Big]=
\e^{\left(f_r(t)\,-\,\frac{1}{2}(\vec{r}_t\cdot \Sigma\cdot \vec{r}_t)\right)}=
\tr\Big[\rho_{\Sigma(t)}\, W(\vec{r}\,)\Big]\ ,
\end{equation}
with the time-dependent covariance matrix $\Sigma(t)$ explicitly given by:
\begin{equation}
\Sigma(t)=\Sigma_\beta\,-\,\mathcal{M}_t\,\Sigma_\beta\,\mathcal{M}_t{}^T+
\,\mathcal{M}_t\,\Sigma\,\mathcal{M}_t{}^T\ .
\label{2.26}
\end{equation}

As already remarked, the mesoscopic state $\Omega$, with density matrix $\rho_\Omega$, defined in \hbox{\sl Theorem 1}, is Gaussian with covariance matrix $\Sigma_\beta$; as the microscopic state $\rho^{(N)}$ is invariant under the local dissipative dynamics $\Phi^{(N)}_t$, $\rho_\Omega$
results invariant under the mesoscopic dissipative dynamics $\widetilde{\Phi}_t$, {\it i.e.} $\Sigma(t)=\Sigma_\beta$.

We are now ready to discuss the entanglement properties of our two-chain system at the
mesoscopic level, using the collective variables $(\hat X_1,\hat P_1,\hat X_2,\hat P_2,\hat X_3,\hat P_3)$
introduced in (\ref{1.40}). Actually, since only $\hat X_1$, $\hat P_1$ represent operators pertaining
to the first chain, and $\hat X_2$, $\hat P_2$ only to the second, while $\hat X_3$, $\hat P_3$ are mixed ones belonging to both chains, one should focus on the first two couples. This means that given a mesoscopic 
Gaussian state $\rho_\Sigma$ for the system,
one should trace out the third degrees of freedom, thus obtaining a reduced state $\hat{\rho}_\Sigma$,
involving only the first two modes, but still in Gaussian form. One can easily check that
the corresponding, reduced, two-mode covariance matrix
$\widehat{\Sigma}$ can be simply obtained from the general one $\Sigma$ by deleting from it rows and columns
involving the third degree of freedom. The resulting covariance $\widehat{\Sigma}$ 
is a $4\times 4$ matrix, that can be organized in $2\times 2$ blocks:
\begin{equation}
\widehat{\Sigma}=
\left(
\begin{array}{c|c}
{\Sigma}_1  & {\Sigma}_c \\ \hline
{\Sigma_c{}^\dagger} & {\Sigma}_2
\end{array}
\right)\ .
\label{2.27}
\end{equation}
The entanglement content of any two-mode Gaussian state can be easily studied; indeed, it turns out
that the operation of partial transposition is an exhaustive entanglement witness~\cite{Simon},
offering in addition a way to quantify quantum correlations. It can be conveniently formulated
in terms of the previous decomposition of the covariance matrix \cite{Souza,Isar}. By defining the four quantities:
\begin{equation}
I_1=\det(\Sigma_1)\ ,\qquad I_2=\det(\Sigma_2)\qquad I_3=\det(\Sigma_c)\ ,\qquad
I_4=\tr\Big(\Sigma_1\sigma_3\Sigma_c\sigma_3\Sigma_2\sigma_3\Sigma_c^{\dagger}\sigma_3\Big)\ ,
\label{2.28}
\end{equation}
with $\sigma_3$ the third Pauli matrix, the necessary and sufficient condition for a state to be separable is:
\begin{equation}
\mathcal{S}\equiv I_1I_2+\bigg(\frac{1}{4}-|I_3|\bigg)^2-I_4-\frac{(I_1+I_2)}{4}\ge0\ .
\label{2.29}
\end{equation}
Further, the amount of entanglement in two-mode Gaussian states can be measured 
through the so-called logarithmic negativity of the state: 
\begin{equation}
E=\max\left\{0,-\frac{1}{2}\log_2\left(4\, {\cal I}\right)\right\}\ ,\\
\label{2.30}
\end{equation}
where
\begin{equation}
{\cal I}=\frac{I_1+I_2}{2}-I_3-\bigg(\left[\frac{I_1+I_2}{2}-I_3\right]^2-(I_1I_2+I_3^2-I_4)\bigg)^{1/2}\ .
\label{2.31}
\end{equation}
These results will be now used to analyze the dynamical behaviour of the quantum correlations
between the two chains while following the mesoscopic time evolution $\Phi_t$.

\section{Environment induced mesoscopic entanglement}

Using the previous results, we will now show that
the two chains can get entangled at the mesoscopic level
through the dissipative dynamics $\Phi_t$, without any
direct interaction between them; further, we shall
investigate the behaviour of this bath-generated entanglement in the course of time 
and of its dependence on the dissipative coupling $\lambda$ and the temperature of the initial state. 

By {\it mesoscopic entanglement} we mean the existence of mesoscopic states carrying non-local, 
quantum correlations among the collective operators pertaining to different chains. 
More precisely, we shall focus on the operators $\hat X_1$, $\hat P_1$ and 
$\hat X_2$, $\hat P_2$, that, as already observed, 
are collective degrees of freedom attached to the first, second chain, respectively. 
We shall then study the dynamics of two-mode Gaussian states $\rho_{\widehat\Sigma}$ 
obtained by tracing a full three-mode Gaussian state $\rho_\Sigma$ over the variables $\hat X_3$, $\hat P_3$.

Since $\rho_\Omega$ and thus $\rho_{\widehat\Sigma}$ are time invariant, 
in order to have a non-trivial evolution,
as initial state of the system we shall take a deformation of the mesoscopic state
$\rho_\Omega$, obtained by applying to it suitable squeezing operators,
\begin{equation}
\rho_\Omega^{(k)}=S_1(k)\, S_2(k)\, \rho_\Omega\,  S_2(k)^\dagger\, S_1(k)^\dagger\ ,
\label{3.1}
\end{equation}
involving only the first two relevant modes:%
\footnote{For simplicity we take identical squeezing operations for the two modes,
depending on the real squeezing parameter $k$.} 
\begin{equation}
S_i(k)=\e^{ik\big( \hat X_i \hat P_i + \hat P_i \hat X_i \big)}\ , \quad i=1,2\ .
\label{3.2}
\end{equation}
Notice that the modes are not mixed by the squeezing operation, so that the resulting three-mode state
$\rho_\Omega^{(k)}$ is still a separable Gaussian state. After tracing over the third mode,
the corresponding two-mode covariance matrix $\widehat{\Sigma}_\Omega^{(k)}$ takes the form (\ref{2.27}), with
\begin{equation}
\Sigma_1=\Sigma_2=\frac{1+\eta^2}{4\eta}
\begin{pmatrix}
\e^{4k} & 0\\
0 & \e^{-4k}
\end{pmatrix}\ ,
\qquad
\Sigma_c=0\ ,
\label{3.3}
\end{equation}
showing explicitly the absence of correlations between the two modes.
Under the mesoscopic dynamics $\Phi_t$ obtained in the previous Section, 
the initial state $\rho_\Omega^{(k)}$ will be mapped
into the new Gaussian state $\rho_\Omega^{(k)}(t)$, whose covariance matrix will evolve
according to the law (\ref{2.26}). Restricting to the first two modes, one explicitly finds that
the reduced covariance becomes:
\begin{equation}
\widehat{\Sigma}_\Omega^{(k)}(t)=
\begin{pmatrix}
\mathcal{H}(t)&{ 0}\\
{ 0}&\mathcal{H}(t)
\end{pmatrix} \begin{pmatrix}
A(t)& B(t) \\ 
B(t) & A(t)
\end{pmatrix}\begin{pmatrix}
\mathcal{H}(t)&{ 0}\\
{ 0}&\mathcal{H}(t)
\end{pmatrix}^T\ ,
\label{3.4}
\end{equation}
where
\begin{equation}
\mathcal{H}(t)=
\begin{pmatrix}
\cos\left(2\omega\, t\right)&\sin\left(2\omega\, t\right)\\
-\sin\left(2\omega\, t\right)&\cos\left(2\omega\, t\right)
\end{pmatrix}
\label{3.5}
\end{equation}
accounts for the Hamiltonian part of the dynamics, that, as already observed, does not mix
with the dissipative one, whose contribution is instead encoded in:
\begin{eqnarray*}
&&A(t)=\frac{1+\eta^2}{4 \eta}\left\{\frac{\e^{-\frac{4t\eta}{1+\eta}}}{4}\bigg[3+
\cosh\bigg( \frac{4\eta\lambda t}{1+\eta}\bigg)\bigg] 
\begin{pmatrix}
e^{4k}-1 & 0\\
0 & e^{-4k}-1\\
\end{pmatrix}+{\bf 1}_2\right\}\ ,\\
&&B(t)=\frac{1+\eta^2}{4 \eta}\left\{\frac{\e^{-\frac{4t\eta}{1+\eta}}}{4}\bigg[
\cosh\bigg( \frac{4\eta\lambda t}{1+\eta}\bigg)-1\bigg] 
\begin{pmatrix}
e^{4k}-1 & 0\\
0 & e^{-4k}-1\\
\end{pmatrix}\right\}\ .
\end{eqnarray*}
From these results, one can now study the entanglement content of the two-chain state
by analyzing the behaviour of the logarithmic negativity $E$ introduced in (\ref{2.30});
indeed, by defining:
\begin{equation}
\mathcal{E}(t)=-\frac{1}{2}\log_2\big(4\, {\cal I}(t)\big)
\label{3.6}
\end{equation}
with ${\cal I}$ as in (\ref{2.31}), one explicitly finds:
\begin{equation}
\mathcal{E}(t)=\left(\frac{1+\eta ^2}{4 \eta }\right)^2 \e^{-4\big(k+\frac{2 \eta  t}{1+\eta}\big)}
\bigg(\e^{4k}+\e^{\frac{4 \eta  t}{1+\eta}}-1\bigg) 
\left[\e^{4\left(k+\frac{\eta  t}{1+ \eta}\right)}-\left(\e^{4 k}-1\right) \cosh
   ^2\left(\frac{2 \eta\lambda  t}{1+\eta}\right)\right]\ .
\label{3.7}
\end{equation}

As clearly shown by the figures below, reporting the behavior in time of $E$,
the dissipative, mesoscopic dynamics $\Phi_t$ can indeed generate
quantum correlations starting from a completely separable initial state,
provided a nonvanishing squeezing parameter $k$ is chosen. Since the Hamiltonian does not contain coupling terms
and the dynamics it generates completely
decouples giving no contribution to $E$, entanglement between the
two chains is generated at the mesoscopic, collective level by the purely noisy
action of the environment in which the two chains are immersed.

\begin{figure}[h!]
\center\includegraphics[scale=0.45]{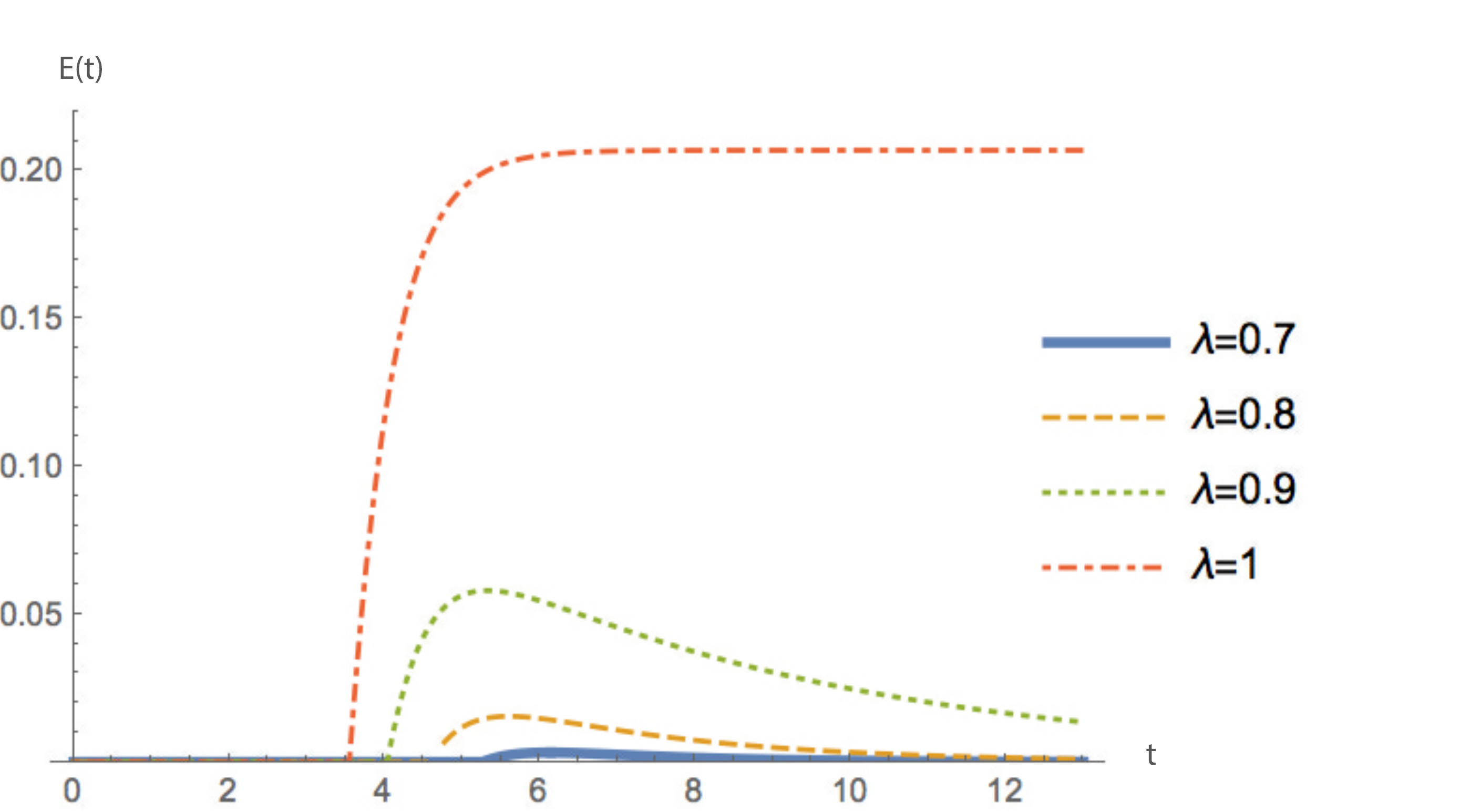}
\caption{\small Behaviour in time of the logarithmic negativity $E$ for different values of the dissipative parameter 
$\lambda$, at fixed squeezing parameter, $k=1$ and temperature, $T=0.1$.}
\label{Fig1}
\end{figure}

\begin{figure}[h!]
\center\includegraphics[scale=0.45]{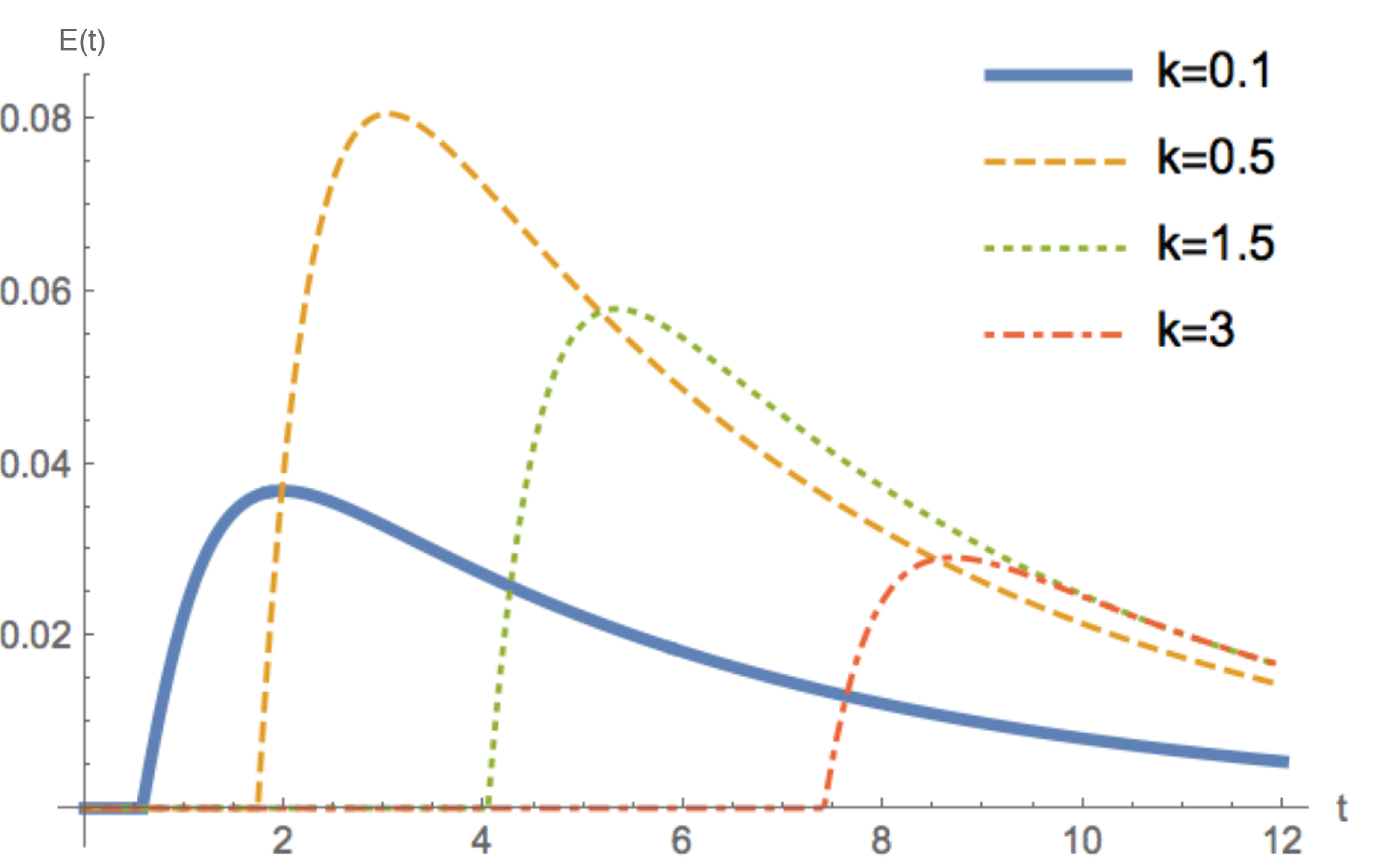}
\caption{\small Behaviour in time of the logarithmic negativity $E$ for different values of the squeezing parameter $k$,
at fixed dissipative parameter, $\lambda=0.9$, and temperature, $T=0.1$.}
\label{Fig2}
\end{figure}

\begin{figure}[h!]
\center\includegraphics[scale=0.45]{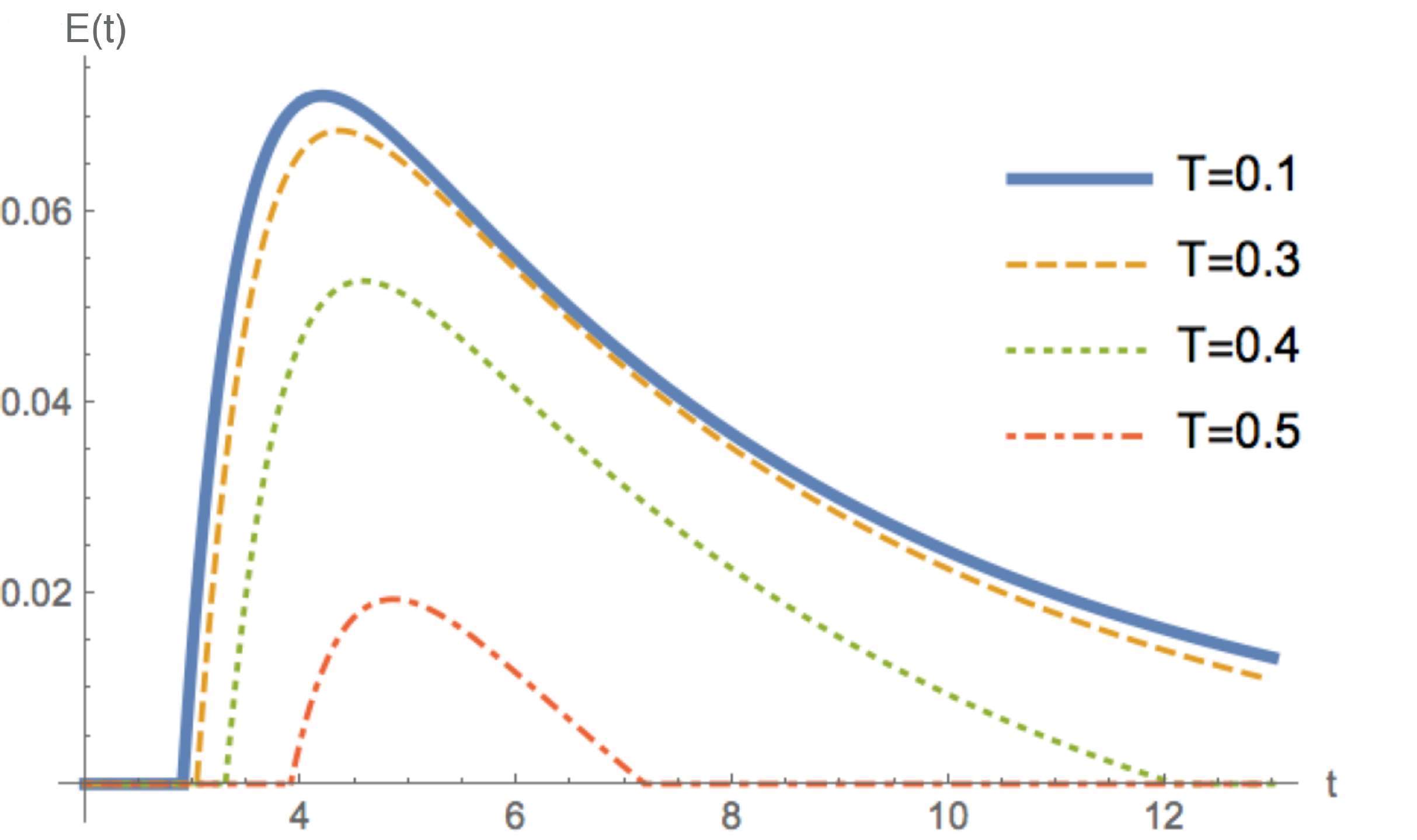}
\caption{\small Behaviour in time of the logarithmic negativity $E$ for different values of the temperature $T$, at fixed dissipative, $\lambda=0.9$, and squeezing parameter, $k=1$.}
\label{Fig3}
\end{figure}

The behaviour in time of the created entanglement depends on the parameter
$\lambda$, measuring the coupling of the system with the environment,
the initial squeezing parameter $k$ and the the bath temperature $T=1/\beta$,
(through the parameter $\eta$).
One sees that the generated entanglement increases 
as the dissipative coupling $\lambda$ gets larger ({\it cf.} Figure 1),
while a non-zero entanglement appears earlier in time.

Also the amount of squeezing plays an essential role; 
while a non-vanishing squeezing appears necessary to create quantum correlations, 
too much squeezing decreases the maximum value of $E$ ({\it cf.} Figure 2). Squeezing also influences the time at which it is
first generated. Further, for fixed $T$ and $\lambda$, 
there is a value of the squeezing parameter $k$ allowing for a maximal value of $E$.

Finally, the effect of the temperature is displayed in Figure 3, for fixed dissipative and squeezing
parameters. One sees that
increasing the temperature, the maximum of the logarithmic negativity $E$ decreases, 
indicating that there exists a value of the temperature above which no entanglement is possible
(see also Figure 6(a) below).

\begin{figure}[h!]
\center\includegraphics[scale=0.45]{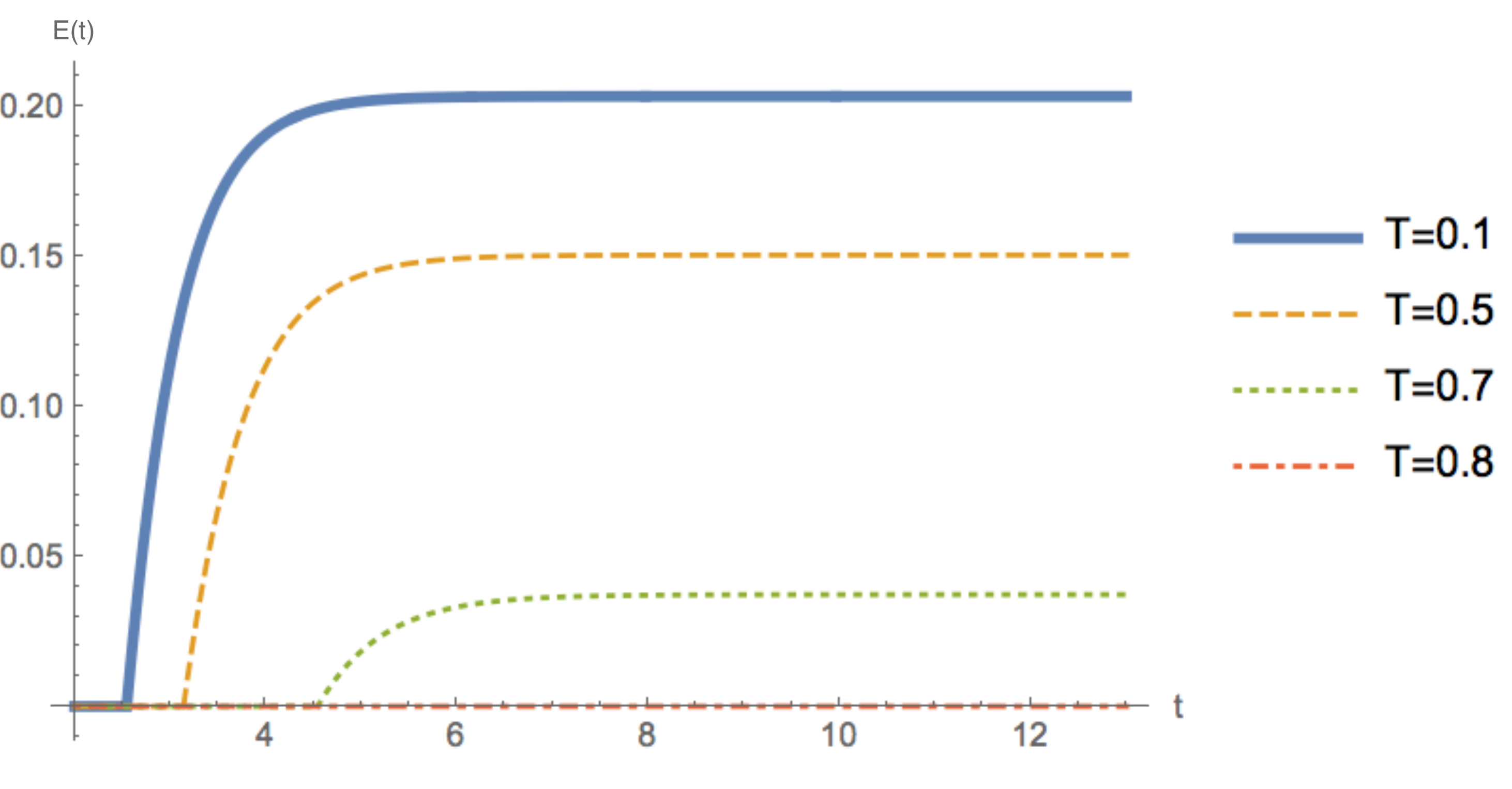}
\caption{\small Behaviour in time of the logarithmic negativity $E$ for different values of the temperature $T$,
at fixed dissipative, $\lambda=1$, and squeezing parameter, $k=1$. Entanglement rapidly reaches an asymptotic 
nonvanishing value even for nonvanishing temperature.}
\label{Fig4}
\end{figure}

\begin{figure}[h!]
\center\includegraphics[scale=0.45]{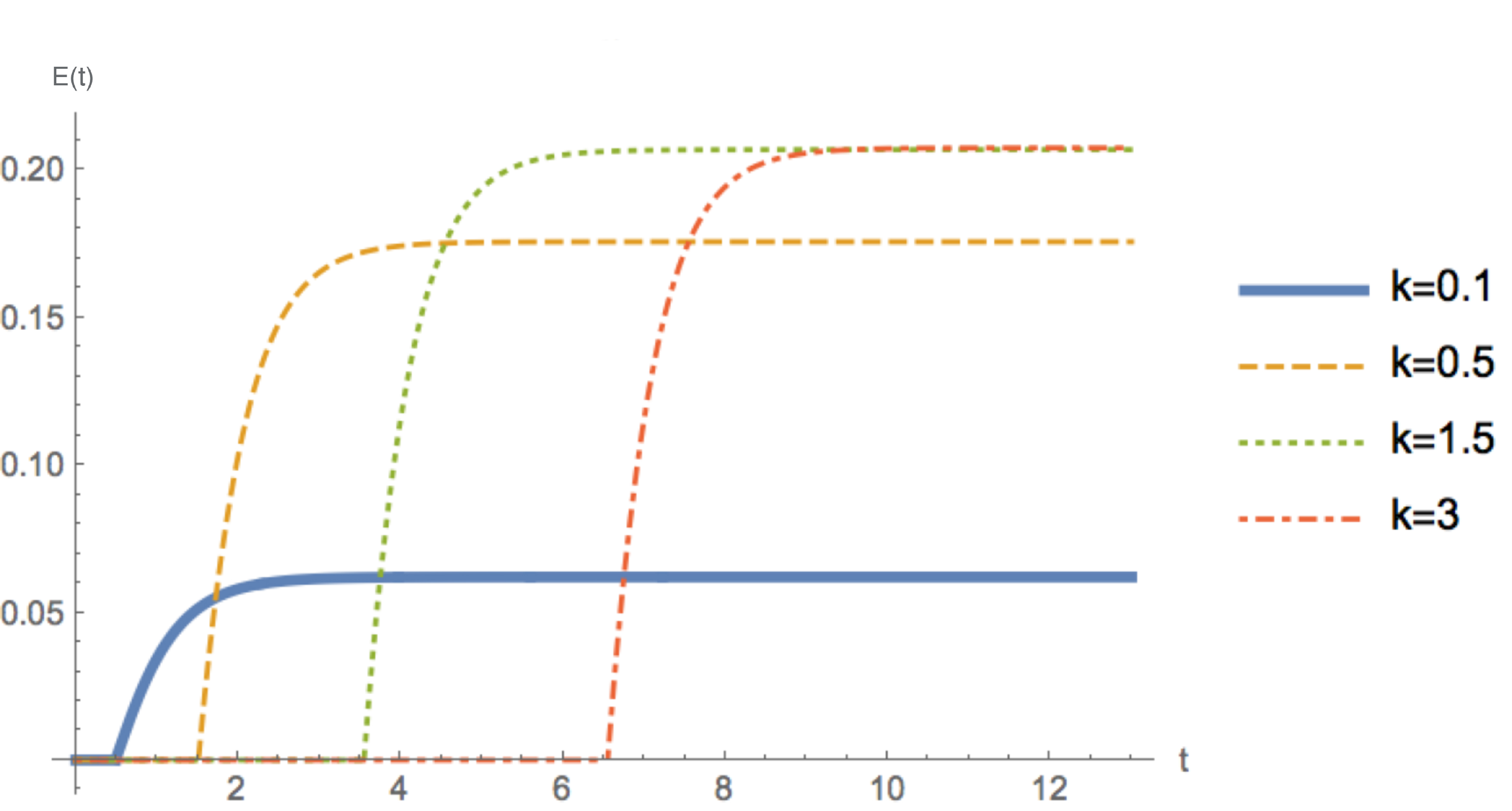}
\caption{\small Behaviour in time of the logarithmic negativity $E$ for different values of the squeezing parameter $k$,
at fixed dissipative parameter, $\lambda=1$, and temperature, $T=0.1$. Notice that, although a nonvanishing 
squeezing is needed for generating entanglement, high values of $k$ do not in general correspond to
a larger asymptotic entanglement.}
\label{Fig4}
\end{figure}

In addition, the time behaviour of the logarithmic negativity $E$ 
shows two further interesting phenomena, the
so-called ``sudden birth'' and ``sudden death'' of entanglement \cite{Eberly}, {\it i.e.} the sudden generation
of entanglement only after a finite time since the starting of the dynamics, 
and the abrupt vanishing of it at a later, finite time. These two effects can be analyzed in detail 
by looking at the explicit expression of ${\cal E}$ in (\ref{3.7}).

In order to study the phenomenon of sudden birth of entanglement, one has to analyze the behaviour
of the logarithmic negativity $E$ in a right neighborhood of $t=0$.
From (\ref{3.7}), one sees that for small times
$\mathcal{E}(t)$ is always negative, for all values of the temperature and squeezing parameter. 
Being $E=\,0$ at $t=\,0$, this implies that it remains so also for small times.
In other terms, a finite time delay is necessary
before quantum correlations can start to be generated by the dissipative dynamics.

To analyze the phenomenon of sudden death, one should instead look at the behavior of $\mathcal{E}(t)$
for large times:
\begin{equation}
\lim_{t \to \infty}\mathcal{E}(t) = \begin{cases}
 -\frac{1}{2}\log\left(\frac{\left(3+e^{-4k}\right)\left(1+\eta^2\right)^2}{16\eta^2}\right) , & \mbox{if } \lambda=1\ , \\
  -\frac{1}{2}\log\left(\frac{\left(1+\eta^2\right)^2}{4\,\eta^2}\right), & \mbox{if } \lambda<1\phantom{\Bigg|}\ .
  \end{cases}
\label{3.8}
\end{equation}
Let us examine first the case $\lambda<1$. In this situation, $\mathcal{E}$ has an asymptotic value
depending only on the temperature $T$ through the parameter $\eta$. For $T>0$, {\it i.e.} $\eta<1$,
this value is negative;
therefore, since entanglement has been created ($k>0$), there must exist a finite time $t=t_0$ at which quantum
correlations vanish, thus giving rise to the phenomenon of sudden death of entanglement.
Further, this time $t_0$ becomes larger and larger as the value of $\lambda$ increases 
(see again Figure~1).
 
Actually, when $\lambda$ reaches its maximum, $\lambda=1$, a non-vanishing asymptotic entanglement
is possible, provided the bath temperature is not too high. This behaviour is clearly shown
in Figure~4 and Figure~5.
This is an important result; it gives the possibility of preparing 
a bipartite many-body quantum systems in a mesoscopic entangled state
by means of the dissipative action of an engineered environment: through a purely mixing mechanism
such a bath produces and protects 
quantum correlations for times becoming the longer, the closer the paramater $\lambda$
is to the value one. This may be important in actual experimental applications,
where achieving exactly $\lambda=1$ might be difficult in practice.

This result is further illustrated by Figure 6(b), where the points in the $(k,T)$ plane 
with non-vanishing large-time mesoscopic entanglement are highlighted. This figure shows two regions,  
a darker one associated with a non-vanishing asymptotic value of $E$ and a brighter one with
vanishing asymptotic value of $E$ and therefore no entanglement. 
The line separating the two regions determines the ``critical temperature'' $T_c$, 
above which entanglement among the two chains is not possible, 
as a function of the squeezing parameter; it is defined implicitly by the condition 
$\lim_{t\to\infty}E(k,T)=0$.

\begin{figure}[h]
 \centering
 \subfigure[]
   {\includegraphics[scale=0.33]{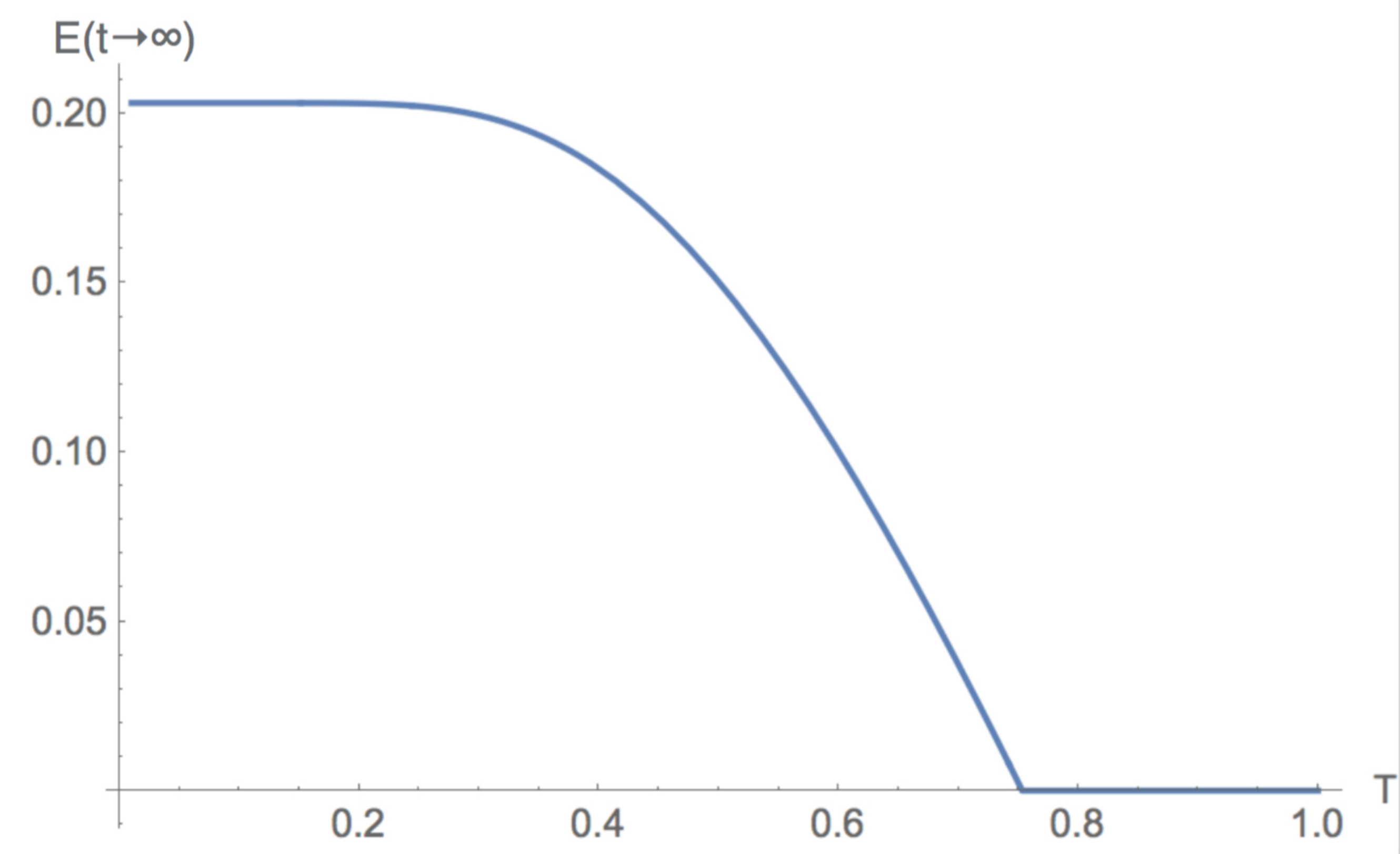}}
\hskip 1 cm 
\subfigure[]
   {\includegraphics[scale=0.45]{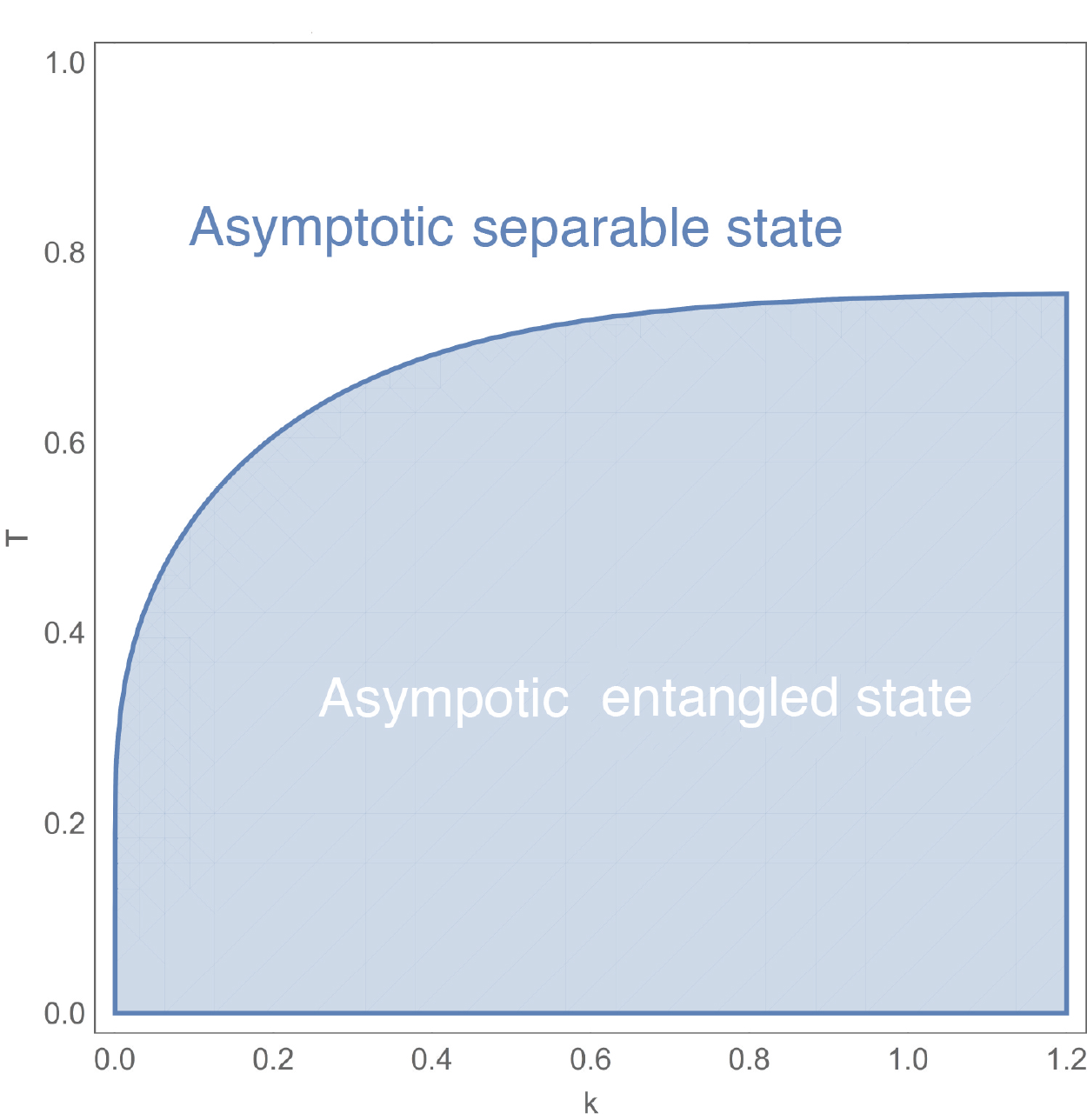}}
 \caption{\small On the left: asymptotic value of the logarithmic negativity $E$ as a function of the temperature $T$
for fixed values of the squeezing $k=1$ and dissipative $\lambda=1$ parameters. 
There clearly exists a critical temperature ($T_c \simeq 0.75$) below which one has a non vanishing asymptotic entanglement. On the right: $(k,T)$-parameter space; the line separates the regions 
in which entanglement is non vanishing and zero, respectively.}
\label{Fig5}
 \end{figure}

\section{Outlook}

The description of many-body systems, {\it i.e.} of systems made of a large number $N$ of elementary constituents, 
generally involves the analysis of collective observables, accounting for all their degrees of freedom.
Mean-field operators are typical examples of such observables: they are algebraic means of single particle
observables, as the case of mean magnetization in spin systems. These quantities scale as $1/N$ as the
number of constituents increase, thus behaving as ``classical'' observables in the thermodynamic,
large-$N$ limit.

On the contrary, fluctuation operators, defined in analogy with classical stochastic theory as deviations from the mean,
retain a quantum character even in the large-$N$ limit: they are a different class of collective observables
scaling as $1/\sqrt{N}$. The algebra they form turns out to be in general non-commutative and always of bosonic type,
allowing probing the quantum character of the many-body system at the mesoscopic level, in between
the microscopic single-particle world and the classical macroscopic regime.

Within this general framework, we have discussed the quantum dynamics of fluctuations 
of a system composed by two independent chains
of free oscillators, both immersed in a weakly-coupled external bath. The total system is therefore open,
so that noise and dissipation ought to occur. Nevertheless, despite the decohering and noisy effects
induced by the bath, the two chains can get entangled at the mesoscopic scale through 
a purely mixing-enhancing mechanism, thanks to the properties of the emergent open dynamics of fluctuations.

We have analyzed in detail the behaviour of such environment generated, collective entanglement
and its dependence on the initial system temperature and the strength of the coupling between system and bath.
Despite its inevitable dissipative action, the environment 
can nevertheless sustain non vanishing collective quantum correlations
among the two chains for asymptotically large times, even at nonvanishing temperatures. 
This is a relevant
result, since so far a nonvanishing asymptotic collective entanglement for many-body fluctuation observables
has been obtained only at zero temperature. 

The existence of an asymptotic mesoscopic state that, differently from
the separable stationary thermal state, is entangled, reveals that a
rich convex set of asymptotic states is enforced by the structure of the
Kossakowski-Lindblad generator \cite{Frigerio}, and specific protocols have been proposed
to prepare predefined entangled states via the action of suitably
engineered environments \cite{Kraus}-\cite{Muschik2}.
Clearly, the structure of the generator
depends on the choice of microscopic observables whose mesoscopic
fluctutations have been proved to become entangled; therefore, the
practical availability of this quantum resource in general depends on
the actual experimental accessibility of many-body observables scaling
with the inverse square root of the number of particles. This problem,
relevant for specific applications, goes beyond the main scope of this paper which
aims at showing how quantum correlations may occur in many-body contexts
despite the high number of constituents generically expected only to lead 
to a classical behaviour.

Indeed, we expect our results to be of interest in experiments involving
spin-like and optomechanical systems, or ultra-cold gases trapped in optical lattices and, more in general,
in all instances where a coherent quantum behaviour is expected to emerge at the mesoscopic level; in particular, 
the possibility of entangling many-body systems through a purely mixing mechanism at nonzero temperature
will surely reinforce their use in quantum information and quantum communication.

\vskip 2cm

\section{Appendix}

We collect in this Appendix the proofs of the results presented in the main text, which, for their
technical character, would hamper the presentation.

\subsection{Algebra of mean-field operators}

In order to prove that the algebra of mean-field operators $X^{(N)}$, as defined in (\ref{1.16}),
is commutative, it is convenient to work with exponentials of the form $\e^{i X^{(N)}}$.
Then the following result holds:

\begin{proposition}
Given a set of single-site observables, $X_1$, $X_2$,\ldots, $X_n$ and the state $\rho^{(N)}$ in (\ref{1.3}),  
their corresponding mean-field averages $X^{(N)}_j$, defined as in (\ref{1.16}), are such that 
\begin{equation}
\lim_{N\to\infty}\bigg\langle\prod_{j=1}^n  \e^{i X^{(N)}_j}\bigg\rangle_N=
\exp\left(i\sum_{j=1}^n \expec{X_j}\right)=
\lim_{N\to\infty}\bigg\langle \e^{i\sum_{j=1}^n X^{(N)}_j}\bigg\rangle_N\ .
\label{A.1}
\end{equation}
\end{proposition}

\begin{proof}
We first prove that, given a single-site observable $X$, one has 
\begin{equation}
\lim_{N\to\infty}\big\langle A^{(N)}\, \e^{i X^{(N)}}\big\rangle_N = 
\e^{i\expec{X}}\lim_{N\to\infty}\expec{A^{(N)}}_N\ ,
\label{A.2}
\end{equation}
for any bounded operator $A^{(N)}$ in the oscillator algebra ${\cal A}^{(N)}$, $\|A^{(N)}\|<\infty$, $\forall\, N$. In order to show this, 
let us consider the difference: 
$$
I^{(N)}=\e^{i X^{(N)}}-\e^{i\expec{X}}\ ,
$$ 
that can be conveniently rewritten as:
$$
I^{(N)}=\int_0^1{\rm d}s\, \frac{{\rm d}}{{\rm d}s} \e^{is X^{(N)}}\,e^{i(1-s)\expec{X}}\, ;
$$
evaluating the derivative one gets:
$$
I^{(N)}=i\int_0^1{\rm d}s\, \e^{is X^{(N)}}\,e^{i(1-s)\expec{X}}\,\left(X^{(N)}-\expec{X}\right)\, .
$$
Using the Cauchy-Schwarz inequality, one can then write:
$$
\Big|\expec{A^{(N)}\, I^{(N)}}_N\Big|\le \|A^{(N)}\|\ \sqrt{\expec{\big(X^{(N)}-\expec{X}\big)^2}_N}\ .
$$
Since $A^{(N)}$ is bounded, one needs to analyze the behaviour of the expectation under the square root;
by expanding the square, one obtains:
$$
\expec{\big(X^{(N)}-\expec{X}\big)^2}_N=\frac{1}{N^2}\sum_{j,k=1}^N\Big(\expec{X^{[j]} X^{[k]}}-\expec{X}^2\Big)\ ,
$$
and further using (\ref{1.8}) and (\ref{1.9}),
$$
\expec{\big(X^{(N)}-\expec{X}\big)^2}_N=\frac{1}{N}\left(\expec{X^2}-\expec{X}^2\right)\ .
$$
In the large-$N$ limit this expectation is thus vanishing, thus proving (\ref{A.2}) above.

This result can now be used to prove the first equality in (\ref{A.1}). Indeed, one can write:
$$
\bigg\langle\prod_{j=1}^n  \e^{i X^{(N)}_j}\bigg\rangle_N =
\bigg\langle\prod_{j=1}^{n-1}  \e^{i X^{(N)}_j}\ \e^{i X^{(N)}_n}\bigg\rangle_N\ ,
$$
and recalling that unitary operators are bounded, (\ref{A.2}) implies:
$$
\lim_{N\to\infty}\bigg\langle\prod_{j=1}^n  \e^{i X^{(N)}_j}\bigg\rangle_N =
\e^{i \expec{X_n}}\ \lim_{N\to\infty}\bigg\langle \prod_{j=1}^{n-1}  \e^{i X^{(N)}_j}\bigg\rangle_N\ .
$$
Repeating this procedure recursively for all exponential factors, one immediately obtains the first equality in (\ref{A.1}).
In addition, using the result (\ref{A.2}) with $A^{(N)}= {\bf 1}$ and $X=\sum_{j=1}^n X_j$, due to the linearity
of the averages, one readily obtains also the second equality in (\ref{A.1}) and thus the proof of the
entire {\sl Proposition}.
\end{proof}
This result shows that the large $N$ limit of mean-field operators $X^{(N)}$ behaves as a multiple of the identity; 
this convergence has to be understood as a converge in distribution, 
similar to the one in the law of large numbers \cite{Feller}; indeed, the expectations of the exponentials involved 
in the {\sl Proposition} are nothing but the characteristic functions of the operators $X^{(N)}$. The result 
$\lim_{N\to\infty}\langle\e^{iX^{(N)}}\rangle_N=\e^{i\expec{X}}$
shows that in the large-$N$ limit $X^{(N)}$ is no longer a quantum random variable, 
rather a deterministic variable, equal to its expectation.

\subsection{Large-$N$ behaviour of ${\cal R}_r^{(N)}$}

In this Section we shall give an estimate for the large-$N$ behaviour of the rest ${\cal R}_r^{(N)}$
appearing in the proof of {\sl Lemma 1} and {\sl Lemma 2}; it can be deduced from a general
result, as expressed by the following {\sl Lemma}.%
\footnote{Using different techniques, the general case is treated in \cite{Goderis1}-\cite{Goderis3}.}

\begin{lemma}
Given a zero-average Gaussian state $\rho$ and an homogeneous polynomial $X$ of degree two in the canonical variables
$\{x_1,p_1,x_2,p_2\}$, the sum
\begin{equation}
{\cal R}_\ell^\delta=\sum_{k=\ell}^\infty\frac{i^k}{k!}\left(\frac{1}{N^{\delta}}\right)^k\big(X-\expec{X}\big)^k\ ,
\qquad \delta>0\ ,\ \ell\in\mathbb{N}\ ,
\label{A.3.1}
\end{equation}
behaves such that 
\begin{equation}
\left|\expec{A^{(N)}\ {\cal R}_\ell^\delta\ B}\right|=O\Big( N^{-\ell\,\delta}\Big)\, ,
\label{A.3.2}
\end{equation}
for $N$ large enough, where $A^{(N)}$ is any bounded operator in the oscillator algebra, 
\hbox{$\|A^{(N)}\|<\infty$}, $\forall\, N$, and  $B$ a monomial of degree $n$ in $\{x_1,p_1,x_2,p_2\}$.
\end{lemma}

\begin{proof}
Using the definition (\ref{A.3.1}) and bounding the modulus of the sum with the sums of the moduli, one can write:
\begin{equation}
\left|\expec{A^{(N)}\ {\cal R}_\ell^\delta\ B}\right|\le \sum_{k=\ell}^{\infty}\frac{1}{k!}
\left(\frac{1}{N^{\delta}}\right)^{k}\left|\expec{A^{(N)}\ \big(X-\expec{X}\big)^k \ B}\right|\ .
\label{A.4}
\end{equation}
Further, using the binomial theorem, the modulus inside the sum can be bounded as follows:
\begin{equation}
\Big|\expec{A^{(N)}\ \big(X-\expec{X}\big)^k\ B}\Big|\le\sum_{m=0}^{k}\binom{k}{m}\Big|\expec{X}^{k-m}\Big|\ 
\Big|\expec{A^{(N)}\ X^m\ B}\Big|\, .
\label{A.5}
\end{equation}
Since by hypothesis $X$ is a polynomial of degree two in the canonical variables, it can be expanded as 
a linear combination of monomials $q_i$ of degree two in $\{x_1,p_1,x_2,p_2\}$,
$X=\sum_{i=1}^d c_i\,  q_i$, with $c_i$ real coefficients. As a consequence, one can then write:
$$
\expec{A^{(N)}\, X^m\,B}=\sum_{i_1,i_2,\dots,i_{m}=1}^d c_{i_1}c_{i_2}\dots c_{i_m}
\expec{A^{(N)}\, q_{i_1}q_{i_2}\dots q_{i_m}\ B}\, .
$$
Since also $B$ is a monomial of degree $n$ in the canonical variables, the entire
product $\mathcal{P}\equiv q_{i_1}q_{i_2}\dots q_{i_m}\ B$ is itself a (not ordered) monomial of
degree $2m+n$. Further, one can bound:
$$
\left|\expec{A^{(N)}) \, q_{i_1}q_{i_2}\dots q_{i_m}\ B}\right|=
\left|\expec{A^{(N)}) \, {\cal P}}\right|
\le \|A^{(N)}\|\sqrt{\expec{\mathcal{P}^\dagger \mathcal{P}}}.
$$
Now, $\mathcal{P}^\dagger \mathcal{P}$ is actually a product of $2(2m+n)$ elements of the set
$\{x_1,p_1,x_2,p_2\}$; therefore, its expectation on the Gaussian state $\rho$ can be
expressed in terms of sums of products of two-point correlation functions through Wick's theorem.
By calling $M$ the maximum of the modulus of all two-point functions, one can then estimate:
$$
\expec{ \mathcal{P}^\dagger \mathcal{P} }\le (2(2m+n)-1)!!\ M^{2m+n}\ ,
$$
since Wick's decomposition involve precisely $(2(2m+n)-1)!!$ terms.
Collecting these results, one can now write:
$$
\left|\expec{A^{(N)}\, X^m\,B}\right|\le\|A^{(N)}\|\, M^{n/2}\, C^m\, \sqrt{(2(2m+n)-1)!!}\ , \qquad C=d\, M\, |c|\, ,
$$
with $|c|=\max\left\{|c_i|\right\}$. 
Inserting this in (\ref{A.5}), and recalling (\ref{A.4}), one can now write:
\begin{equation}
\begin{split}
\left|\expec{A^{(N)}\ {\cal R}_\ell^\delta\ B}\right|&\le \|A^{(N)}\|\, M^{n/2}\sum_{k=\ell}^\infty\frac{C^{k}}{k!}\left(\frac{1}{N^{\delta}}\right)^{k}\sum_{m=0}^{k}\binom{k}{m}\sqrt{(2(2m+n)-1])!}\left|\langle X\rangle\right|^{k-m}\\
&\le \|A^{(N)}\|\, M^{n/2}
\sum_{k=\ell}^\infty\frac{(C')^{k}}{k!}\left(\frac{1}{N^{\delta}}\right)^{k}\sqrt{[2(2k+n)-1]!!}\ ,
\end{split}
\end{equation}
where $C'=C(1+\langle X\rangle)$.
At this point, in order to prove the large-$N$ behaviour stated in (\ref{A.3.2}), 
it is sufficient to show that 
$$
N^{\ell\delta}\left|\expec{A^{(N)}\ {\cal R}_\ell^\delta\ B}\right|<\infty\, ,
$$
or equivalently that the series
$$
\sum_{k=\ell}^\infty\frac{(C')^{k}}{k!}\left(\frac{1}{N^{\delta}}\right)^{k-\ell}\sqrt{(2(2k+n)-1)!!}\ ,
$$
converges; using the ratio test, this is ensured by the condition $4C'/N^\delta <1$, {\it i.e.}
for $N$ large enough.
\end{proof}

\subsection{Weyl operators}

In this Section we consider the large-$N$ limit of the Weyl-like operators $W^{(N)}(\vec{r}\,)$
defined in (\ref{1.26}) and show that they behave as true Weyl operators. As discussed in the main text,
this is guaranteed by the following {\sl Lemma}:

\bigskip
\noindent
{\bf Lemma 2.\ }{\it Given the state $\rho^{(N)}$ in (\ref{1.3}), and two distinct single-site operators $X_{r_1}$, $X_{r_2}$ belonging
to the real linear span $\mathcal{X}$ in (\ref{1.24}), one has
$$
\lim_{N\to\infty}\expec{A^{(N)}\,\left(W^{(N)}(\vec{r_1})\, W^{(N)}(\vec{r_2})-
W^{(N)}(\vec{r_1}+\vec{r_2})\ \e^{-\frac{1}{2}\expec{\big[X_{r_1}, X_{r_2}\big]}}\right)}_N=0\ ,
$$
for any bounded element $A^{(N)}$
in the oscillator algebra ${\cal A}^{(N)}$, {\it i.e.} $\left\|A^{(N)}\right\|<\infty$, $\forall\,N$.}
\begin{proof}
Let us define 
$$
\Delta^{(N)}=W^{(N)}(\vec{r_1})\, W^{(N)}(\vec{r_2})-
W^{(N)}(\vec{r_1}+\vec{r_2})\ \e^{-\frac{1}{2}\expec{\big[X_{r_1}, X_{r_2}\big]}}\ ,
$$
and then focus on the $\rho^{(N)}$ expectation value in the above limit. 
Using the Cauchy-Schwarz inequality, its modulus can be bounded as
\begin{equation}
\big|\expec{A^{(N)}\, \Delta^{(N)}}_N\big|\le 
\left\|A^{(N)}\right\|\sqrt{\expec{\big(\Delta^{(N)}\big)^\dagger \Delta^{(N)}}_N}\ ,
\label{A.7}
\end{equation}
where, explicitly:
$$
\expec{\big(\Delta^{(N)}\big)^\dagger \Delta^{(N)}}_N=
2-2Re\left\{\e^{-\frac{1}{2}\expec{\big[X_{r_1}, X_{r_2}\big]}}\Big\langle{W^{(N)}(\vec{-r_2})\, W^{(N)}(\vec{-r_1})\,
W^{(N)}(\vec{r_1}+\vec{r_2})}\Big\rangle_N \right\}\ .
$$
Further, recalling again the property (\ref{1.7}) of the state $\rho^{(N)}$, one can write:
\begin{eqnarray}
\nonumber
&&\Big\langle{W^{(N)}(\vec{-r_2})\, W^{(N)}(\vec{-r_1})\, W^{(N)}(\vec{r_1}+\vec{r_2})}\Big\rangle_N=\\
&&\hskip 5cm \expec{\e^{-i\frac{X_{r_2}-\expec{X_{r_2}}}{\sqrt{N}}}\, \e^{-i\frac{X_{r_1}-\expec{X_{r_1}}}{\sqrt{N}}}\,
\e^{i\frac{X_{r_1}+X_{r_2}-\expec{X_{r_1}+X_{r_2}}}{\sqrt{N}}}}^N\ .
\label{1.32-1}
\end{eqnarray}
We need now study the large-$N$ limit of the single site expectation on the r.h.s. of this relation.
As in the proof {\sl Lemma 1} in Section 2, any single-site exponential 
$\e^{i\frac{X_r-\expec{X_r}}{\sqrt{N}}}$ can be expanded as in (\ref{1.28}):
$$
\e^{i\frac{X_r-\expec{X_r}}{\sqrt{N}}}=Q^{(N)}(X_r)+{\cal R}_r^{(N)}\ ,
$$
with 
\begin{equation*}
Q^{(N)}(X_r)=1+\frac{i}{\sqrt{N}}\big(X_r-\expec{X_r}\big)+\frac{i^2}{2N}\big(X_r-\expec{X_r}\big)^2\ ,
\end{equation*}
and ${\cal R}_r^{(N)}$ as in (\ref{1.29}). By expanding the rightmost exponential in (\ref{1.32-1}), one can write:
\begin{eqnarray}
\nonumber
&&\Big\langle{W^{(N)}(\vec{-r_2})\, W^{(N)}(\vec{-r_1})\, W^{(N)}(\vec{r_1}+\vec{r_2})}\Big\rangle_N =\\
&&\hskip 4cm
\expec{\e^{-i\frac{X_{r_2}-\expec{X_{r_2}}}{\sqrt{N}}}\, \e^{-i\frac{X_{r_1}-\expec{X_{r_1}}}{\sqrt{N}}}\,
\Big(Q^{(N)}(X_{r_1 + r_2})+{\cal R}_{r_1 + r_2}^{(N)}\Big)}^N\ .
\nonumber
\end{eqnarray}
Using the results of the previous {\sl Lemma 3} in Section 6.2, one shows that for large $N$ one has: 
$$
\Bigg|\expec{\e^{-i\frac{X_{r_2}-\expec{X_{r_2}}}{\sqrt{N}}}\, \e^{-i\frac{X_{r_1}-\expec{X_{r_1}}}{\sqrt{N}}}\,
\, {\cal R}_{r_1 + r_2}^{(N)}}\Bigg|= O(N^{-3/2})\, ,
$$
and therefore
\begin{eqnarray}
\nonumber
&&\Big\langle{W^{(N)}(\vec{-r_2})\, W^{(N)}(\vec{-r_1})\, W^{(N)}(\vec{r_1}+\vec{r_2})}\Big\rangle_N=\\
&&\hskip 4cm 
\expec{\e^{-i\frac{X_{r_2}-\expec{X_{r_2}}}{\sqrt{N}}}\, \e^{-i\frac{X_{r_1}-\expec{X_{r_1}}}{\sqrt{N}}}\,
\, Q^{(N)}(X_{r_1 + r_2})+ O(N^{-3/2})}^N\ .
\nonumber
\end{eqnarray}
Repeating recursively the same procedure also for the remaining two exponentials, by using again the
results of {\sl Lemma 3}, one finally gets:
\begin{eqnarray}
\nonumber
&&\Big\langle{W^{(N)}(\vec{-r_2})\, W^{(N)}(\vec{-r_1})\, W^{(N)}(\vec{r_1}+\vec{r_2})}\Big\rangle_N=\\
&&\hskip 4cm
\Big\langle Q^{(N)}(-X_{r_2})\, Q^{(N)}(-X_{r_1})\, Q^{(N)}(X_{r_1 + r_2})\, +\, O(N^{-3/2})\Big\rangle^N\ .
\nonumber
\end{eqnarray}
Expanding the product of the $Q^{(N)}$'s, and keeping only the lowest terms in $1/N$, one is left with 
$$
\Big\langle{W^{(N)}(\vec{-r_2})\, W^{(N)}(\vec{-r_1})\, W^{(N)}(\vec{r_1}+\vec{r_2})}\Big\rangle_N=
\lim_{N\to\infty}\expec{1+\frac{\big\langle[X_{r_1},X_{r_2}]\big\rangle}{2N}}^N=
\e^{\frac{1}{2}\expec{\big[X_{r_1}, X_{r_2}\big]}}\ ,
$$
showing that 
$$
\lim_{N\to\infty}\expec{\big(\Delta^{(N)}\big)^\dagger \Delta^{(N)}}_N=\, 0\ .
$$
Since $A^{(N)}$ is by assumption a bounded operator for any $N$, from (\ref{A.7}) the thesis of the {\sl Lemma} follows.
\end{proof}

\subsection{Large-$N$ behaviour of $\Phi_t^{(N)}\big[{\cal R}_r^{(N)}]$}

In this Section we shall give an estimate on the dissipative time evolution of the rest
${\cal R}_r^{(N)}$, needed in the proof of {\sl Theorem 2} reported in the next Section. 
In analogy with the discussion
in Section 6.2 above, we shall prove a slightly more general result,
given by the following {\sl Lemma}.

\begin{lemma}
Given a zero-average Gaussian state $\rho$, an homogeneous polynomial $X$ of degree two in the canonical variables
$\{x_1,p_1,x_2,p_2\}$ and a generator $\mathbb{L}$ of a quantum dynamical semigroup, 
at most quadratic in the previous canonical variables, the sum
\begin{equation}
{\cal R}_\ell^\delta=\sum_{k=\ell}^\infty\frac{i^k}{k!}\left(\frac{1}{N^{\delta}}\right)^k\big(X-\expec{X}\big)^k\ ,
\qquad \delta>0\ ,\ \ell\in\mathbb{N}\ ,
\label{A.9}
\end{equation}
is such that 
\begin{equation}
\left|\expec{\e^{i\frac{Y-\expec{Y}}{\sqrt{N}}}\ \e^{t\mathbb{L}}\left[{\cal R}_\ell^\delta\right]\ B^{(N)}}\right|
=O\Big( N^{-\ell\,\delta}\Big)\ ,
\label{A.10}
\end{equation}
for $N$ large enough, where $Y$ is a real, homogeneous quadratic polynomial in the canonical variables, 
while $B^{(N)}$ is any operator in the oscillator algebra such that  
\hbox{$ \langle B^{(N)}{}^\dagger\, B^{(N)}\rangle <\infty,\,\forall\, N$.}
\end{lemma}

\begin{proof}
Let us start by considering the expectation:
$$
\expec{\e^{i\frac{Y-\expec{Y}}{\sqrt{N}}}\ \e^{t\mathbb{L}}\left[{\cal R}_\ell^\delta\right]\ B^{(N)}}=\sum_{k=\ell}^{\infty}\frac{i^k}{k!}\left(\frac{1}{N^\delta}\right)^k\expec{\e^{i\frac{Y-\expec{Y}}{\sqrt{N}}}
\e^{t\mathbb{L}}\left[\big(X-\expec{X}\big)^k\right]B^{(N)}}\ ,
$$
whose modulus can then be bounded as:
$$
\left|\expec{\e^{i\frac{Y-\expec{Y}}{\sqrt{N}}}\ \e^{t\mathbb{L}}\left[{\cal R}_\ell^\delta\right]\ B^{(N)}}\right|\le\sum_{k=\ell}^{\infty}\frac{1}{k!}\left(\frac{1}{N^\delta}\right)^k\left|\expec{\e^{i\frac{Y-\expec{Y}}{\sqrt{N}}}\e^{t\mathbb{L}}\left[\big(X-\expec{X}\big)^k\right]B^{(N)}}\right|\ .
$$
Let us then focus on the expectation inside the infinite sum; with the help of the binomial theorem, one obtains:
$$
\left|\expec{\e^{i\frac{Y-\expec{Y}}{\sqrt{N}}}\e^{t\mathbb{L}}\left[\left(X-\expec{X}\right)^k\right]B^{(N)}}\right|\le\sum_{m=0}^k\binom{k}{m}\left|\expec{X}^{k-m}\right|\left|\expec{\e^{i\frac{Y-\expec{Y}}{\sqrt{N}}}\e^{t\mathbb{L}}\left[X^m\right]B^{(N)}}\right|\ .
$$
Using the Cauchy-Schwarz inequality and (\ref{2.13-2}), one can further write
$$
\left|\expec{\e^{i\frac{Y-\expec{Y}}{\sqrt{N}}}\e^{t\mathbb{L}}\left[X^m\right]B^{(N)}}\right|\le
\langle B^{(N)}{}^\dagger\, B^{(N)}\rangle
\expec{\e^{i\frac{Y-\expec{Y}}{\sqrt{N}}}\e^{t\mathbb{L}}\left[X^{2m}\right]
\e^{-i\frac{Y-\expec{Y}}{\sqrt{N}}}}^{1/2}\ ,
$$
with $\langle B^{(N)}{}^\dagger\, B^{(N)}\rangle$ bounded for any $N$ by assumption.
Further, recalling that $Y$ is a sum of monomials of degree $2$ and that the evolution $e^{t\mathbb{L}}$ is quasi-free,
the functional 
$$
\expec{e^{i\frac{Y-\expec{Y}}{\sqrt{N}}}e^{t\mathbb{L}}\left[\ \cdot\ \right]e^{-i\frac{Y-\expec{Y}}{\sqrt{N}}}}
\equiv\expec{\ \cdot\ }_{t,\frac{1}{\sqrt{N}}}\, ,
$$
defines a zero-average Gaussian state on the oscillator algebra,
whose covariance matrix is bounded for any $N$ and for any $t$ belonging to a compact interval.
The two-point functions with respect to this Gaussian state are thus bounded and therefore
one can now proceed exactly as in the proof of {\sl Lemma 2} in Section 6.2 above,
obtaining the convergence condition:
$$
\expec{\e^{i\frac{Y-\expec{Y}}{\sqrt{N}}}\e^{t\mathbb{L}}\left[X^{2m}\right]
\e^{-i\frac{Y-\expec{Y}}{\sqrt{N}}}}^{1/2}\le \, C^m\, \sqrt{(4m-1)!!}\ ,
$$
with a suitable constant $C$. Inserting this result in the above chain of inequalities, 
with steps completely analogous to those used in Section 6.2, one finally
derives the following asymptotic behaviour
$$
\left|\expec{\e^{i\frac{Y-\expec{Y}}{\sqrt{N}}}\ \e^{t\mathbb{L}}\left[{\cal R}_\ell^\delta\right]\ B^{(N)})}\right|
=O\Big( N^{-\ell\delta}\Big)\ ,
$$
valid for $N$ large enough.
\end{proof}

\subsection{Proof of Theorem 2}

In this Section we shall give a prove of the mesoscopic limit
$$
m-\lim_{N\to\infty}\Phi^{(N)}_t\Big[W^{(N)}(\vec{r}\,)\Big]=\Phi_t\left[W(\vec{r}\,)\right]\ ,
$$
which is the key result of {\sl Theorem 2}.

\begin{proof}
As explained at the end of Section 2, the above mesoscopic limit actually means
\begin{equation}
\lim_{N\to\infty} \Big\langle W^{(N)}(\vec{r}_1)\,\Phi^{(N)}_t\big[W^{(N)}(\vec{r}\,)\big]\,W^{(N)}(\vec{r}_2)\Big\rangle_N=
\Big\langle W(\vec{r}_1)\,\Phi_t\big[W(\vec{r}\,)\big]\, W(\vec{r}_2)\Big\rangle_\Omega\ ,
\label{2.14}
\end{equation}
for all $\vec{r},\ \vec{r}_{1},\ \vec{r}_{2}\in \mathbb{R}^6$, and this is precisely what needs to be proven.

Let us first consider the r.h.s. of (\ref{2.14}); using the commutation relations (\ref{1.34}) for Weyl operators
and the results of {\sl Lemma 1}, one can rewrite:
\begin{equation}
\Big\langle W(\vec{r}_1)\,\Phi_t\big[W(\vec{r}\,)\big]\, W(\vec{r}_2)\Big\rangle_\Omega = \e^{\frac{1}{2}(Z_t+iY_t)}\ ,
\label{2.15}
\end{equation}
with
\begin{equation}
\begin{split}
&Z_t=(\vec{r}+\vec{r}_1+\vec{r}_2)\cdot \Sigma_\beta \cdot (\vec{r}+\vec{r}_1+\vec{r}_2)
-\vec{r}_t\cdot \mathcal{K}_t \cdot\vec{r}_t\ ,\\
&Y_t=\vec{r}_1\cdot\sigma\cdot\vec{r}_t+\vec{r}_t\cdot\sigma\cdot\vec{r}_2+\vec{r}_1\cdot \sigma\cdot\vec{r}_2\ .
\end{split}
\label{2.16}
\end{equation}
Therefore, in order to prove the Theorem, one should retrieve the same result from the limiting procedure 
on the l.h.s. of (\ref{2.14}).

Recalling the properties (\ref{1.7}) and (\ref{1.9}) of the state $\rho^{(N)}$, and the definitions
(\ref{1.25}) and (\ref{1.26}) of the fluctuations $F^{(N)}(X_r)$ and the corresponding
Weyl-like operators $W^{(N)}(\vec{r}\,)$, one has:
\begin{eqnarray}
\nonumber
&&\lim_{N\to\infty} \Big\langle W^{(N)}(\vec{r}_1)\,\Phi^{(N)}_t\big[W^{(N)}(\vec{r}\,)\big]\,
W^{(N)}(\vec{r}_2)\Big\rangle_N=\\
&&\hskip 2cm \lim_{N\to\infty}\left(\expec{\e^{\frac{i}{\sqrt{N}}\big(X_{r_1}-\langle X_{r_1}\rangle\big)} 
e^{t\mathbb{L}^{(N)}}\left[\e^{\frac{i}{\sqrt{N}}\big(X_{r}-\langle X_{r}\rangle\big)}\right]
\e^{\frac{i}{\sqrt{N}}\big(X_{r_2}-\langle X_{r_2}\rangle\big)}}\right)^N\ ,
\label{2.17}
\end{eqnarray}
where, recalling (\ref{1.23}), (\ref{1.24}), $X_r$, $X_{r_1}$ and $X_{r_2}$ are sums of monomials of degree 2 
in the canonical variables. 
Let us then focus on the single-site expectation on the r.h.s. of (\ref{2.17}).
By expanding the last exponential, one can write:
$$
\e^{i\frac{X_{r_2}-\expec{X_{r_2}}}{\sqrt{N}}}=Q^{(N)}(X_{r_2})+{\cal R}_{r_2}^{(N)}\ ,
$$
with 
\begin{equation}
Q^{(N)}(X_{r_2})=1+\frac{i}{\sqrt{N}}\big(X_{r_2}-\expec{X_{r_2}}\big)+\frac{i^2}{2N}\big(X_{r_2}-\expec{X_{r_2}}\big)^2\ ,
\label{2.18}
\end{equation}
and ${\cal R}_{r_2}^{(N)}$ as in (\ref{1.29}),
\begin{equation}
{\cal R}_{r_2}^{(N)}=\sum_{k=3}^\infty\frac{i^k}{k!\,\big(\sqrt{N}\big)^k}\big(X_{r_2}-\expec{X_{r_2}}\big)^k\, .
\label{2.19}
\end{equation}
Using the results of {\sl Lemma 3} in Section 6.2, for $N$ large, one can then write:
\begin{eqnarray*}
&&\expec{\e^{\frac{i}{\sqrt{N}}\big(X_{r_1}-\langle X_{r_1}\rangle\big)} 
e^{t\mathbb{L}^{(N)}}\left[\e^{\frac{i}{\sqrt{N}}\big(X_{r}-\langle X_{r}\rangle\big)}\right]
\e^{\frac{i}{\sqrt{N}}\big(X_{r_2}-\langle X_{r_2}\rangle\big)}}=\\
&&\hskip 3cm 
\expec{\e^{\frac{i}{\sqrt{N}}\big(X_{r_1}-\langle X_{r_1}\rangle\big)} 
e^{t\mathbb{L}^{(N)}}\left[\e^{\frac{i}{\sqrt{N}}\big(X_{r}-\langle X_{r}\rangle\big)}\right]
\ Q^{(N)}(X_{r_2})} + O\big(N^{-3/2} \big)\ .
\end{eqnarray*}
By expanding as in (\ref{2.18}) also the middle exponential containing $X_r$, one further obtains:
\begin{eqnarray*}
&&\expec{\e^{\frac{i}{\sqrt{N}}\big(X_{r_1}-\langle X_{r_1}\rangle\big)} 
e^{t\mathbb{L}^{(N)}}\left[\e^{\frac{i}{\sqrt{N}}\big(X_{r}-\langle X_{r}\rangle\big)}\right]
\e^{\frac{i}{\sqrt{N}}\big(X_{r_2}-\langle X_{r_2}\rangle\big)}}=\\
&&\hskip 3cm 
\expec{\e^{\frac{i}{\sqrt{N}}\big(X_{r_1}-\langle X_{r_1}\rangle\big)} 
e^{t\mathbb{L}^{(N)}}\left[Q^{(N)}(X_{r})\right]
\ Q^{(N)}(X_{r_2})} + O\big(N^{-3/2} \big)\ ,
\end{eqnarray*}
since, in the large-$N$ limit, by {\sl Lemma 4} in Section 6.4, 
$\Phi^{(N)}_t\big[ {\cal R}_{r}^{(N)}\big]$ 
gives also contributions of order $1/N^{3/2}$ for any $t\geq 0$. Finally, the expansion of the last
exponential containing $X_{r_1}$ yields:
\begin{eqnarray}
\nonumber
&&\lim_{N\to\infty} \Big\langle W^{(N)}(\vec{r}_1)\,\Phi^{(N)}_t\big[W^{(N)}(\vec{r})\big]\,
W^{(N)}(\vec{r}_2)\Big\rangle_N=\\
&&\hskip 2cm \lim_{N\to\infty}\left(\expec{Q^{(N)}(X_{r_1})\ 
e^{t\mathbb{L}^{(N)}}\left[Q^{(N)}(X_{r})\right]\ Q^{(N)}(X_{r_2}) + O\big(N^{-3/2} \big)
}\right)^N\ .
\label{2.20}
\end{eqnarray}
Recalling the result (\ref{2.11}), one further gets:
\begin{equation*}
e^{t\mathbb{L}^{(N)}}\left[X_{r}\right] = \vec{r}\cdot \mathcal{M}_t\cdot \vec{X}\equiv X_r(t)\ ,
\end{equation*}
and therefore
\begin{equation*}
e^{t\mathbb{L}^{(N)}}\left[Q^{(N)}(X_{r})\right] = 
1+\frac{i}{\sqrt{N}}\big(X_{r}(t)-\expec{X_{r}(t)}\big)+
\frac{i^2}{2N}e^{t\mathbb{L}^{(N)}}\left[\big(X_{r}-\expec{X_{r}}\big)^2\right]\ .
\end{equation*}
With the help of this result and recalling the definition (\ref{2.18}),
by using the shorthand notation $\tilde O=O-\expec{O}$, one finds:
\begin{eqnarray}
\nonumber
&&\expec{Q^{(N)}(X_{r_1})\ 
e^{t\mathbb{L}^{(N)}}\left[Q^{(N)}(X_{r})\right]\ Q^{(N)}(X_{r_2}) }=\\
&& \hskip 4 cm 1 +\frac{i}{\sqrt{N}}\expec{ \tilde{X}_{r_1} + \tilde{X}_r(t) + \tilde{X}_{r_2} }\\
&& \hskip 4 cm -\frac{1}{2N} \expec{ \big(\tilde{X}_{r_1}\big)^2 +
e^{t\mathbb{L}^{(N)}}\left[ \big(\tilde{X}_{r}\big)^2\right] +
\big(\tilde{X}_{r_2}\big)^2 } \\
&& \hskip 4 cm -\frac{1}{N} \expec{
\tilde{X}_{r_1}\, \tilde{X}_{r}(t) + \tilde{X}_{r}(t)\, \tilde{X}_{r_2} +\tilde{X}_{r_1}\,\tilde{X}_{r_2} }
+O\big(N^{-3/2}\big)\ ,
\label{2.21}
\end{eqnarray}
where only the significant orders in $1/N$ are kept. In the above expansion, the
terms scaling as $1/\sqrt N$ are clearly identically zero. Further, recalling the definition (\ref{1.30})
for the covariance matrix $\Sigma_\beta$ and the time invariance of the state, one gets:
$$
\expec{ \big(\tilde{X}_{r_1}\big)^2 }=\vec{r}_1\cdot\Sigma_\beta\cdot\vec{r}_1\ ,\qquad
\expec{e^{t\mathbb{L}^{(N)}}\left[ \big(\tilde{X}_{r}\big)^2\right]}=\vec{r}\cdot\Sigma_\beta\cdot\vec{r}\ ,\qquad
\expec{ \big(\tilde{X}_{r_2}\big)^2 }=\vec{r}_2\cdot\Sigma_\beta\cdot\vec{r}_2\ .
$$
In addition, one easily sees that:
\begin{eqnarray*}
&&\expec{\tilde{X}_{r_1}\, \tilde{X}_{r}(t)}=\vec{r}_1\cdot\left(\Sigma_\beta+\frac{i}{2}\sigma\right)\cdot \vec{r}_t\ ,\\  
&&\expec{\tilde{X}_{r}(t)\, \tilde{X}_{r_2}}=\vec{r}_t\cdot\left(\Sigma_\beta+\frac{i}{2}\sigma\right)\cdot \vec{r}_2\ ,\\
&&\expec{\tilde{X}_{r_1}\, \tilde{X}_{r_2}}=\vec{r}_1\cdot\left(\Sigma_\beta+\frac{i}{2}\sigma\right)\cdot \vec{r}_2\ ,
\end{eqnarray*}
which can be obtained by recalling the expectations of the commutator (\ref{1.32}) and anticommutator (\ref{1.30})
of the single-site operators $X_r$ defined in (\ref{1.23}).
Taking into account that $\Sigma_\beta$ is a symmetric matrix, one can recast the r.h.s. of (\ref{2.17}) as:
\begin{equation*}
\lim_{N\to\infty} \Big\langle W^{(N)}(\vec{r}_1)\,\Phi^{(N)}_t\big[W^{(N)}(\vec{r}\,)\big]\,
W^{(N)}(\vec{r}_2)\Big\rangle_N=
\lim_{N\to\infty}\left(1-\frac{Z_t+iY_t}{2N}\right)^N = \e^{\frac{1}{2}(Z_t+iY_t)}\ ,
\end{equation*}
with $Z_t$ and $Y_t$ as in (\ref{2.16}).
\end{proof}

\vskip 2cm


\end{document}